\documentclass[final]{IEEEtran}
\pdfoutput=1

\usepackage{amsthm,amssymb,graphicx,multirow,amsmath,color,amsfonts}
\usepackage[usenames,dvipsnames]{xcolor}
\usepackage[update,prepend]{epstopdf}
\usepackage[noadjust]{cite}
\usepackage[latin1]{inputenc}
\usepackage{tikz}
\usetikzlibrary{arrows,calc}		
\usepackage{pdfpages}
\usepackage{flushend}
\usepackage{tabulary}
\usepackage{multirow}
\usepackage{comment}

\DeclareMathOperator{\MAC}{MAC}

\def\nb0{{\mathbf{0}}}
\def\nb1{{\mathbf{1}}}

\def\ncalA{{\mathcal{A}}}

\def\ncalC{{\mathcal{C}}}

\def\ncalE{{\mathcal{E}}}

\def\ncalL{{\mathcal{L}}}

\def\ncalN{{\mathcal{N}}}

\def\ncalR{{\mathcal{R}}}
\def\ncalS{{\mathcal{S}}}
\def\ncalT{{\mathcal{T}}}

\def\ncalX{{\mathcal{X}}}

\def\nbbE{{\mathbb{E}}}

\def\nbbP{{\mathbb{P}}}

\def\nbbZ{{\mathbb{Z}}}

\newtheorem{lemma}{Lemma}

\newtheorem{definition}{Definition}	

\newtheorem{theorem}{Theorem}

\newtheorem{assumption}{Assumption}
\newtheorem{example}{Example}
\newtheorem{remark}{Remark}
	
\def\ntaus{{\tau_s}}				
\def\ntaum{{\tau_m}}				
\def\nKslot{{K^{(\tau_s)}}}   		
\def\nKsslot{{K^{(\tau_s)}_s}}		

\def\nKmslot{{K^{(\tau_m)}}}		
\def\nKsmslot{{K^{(\tau_m)}_s}}	

\def\ntKmslot{{\overline{K}^{(\tau_m)}}}		
\def\ntKsmslot{{\overline{K}^{(\tau_m)}_s}}	

\def\ntKmslots{{\overline{K}}}		
\def\ntKsmslots{{\overline{K}_s}}	
 
\def\calN{\mathcal{N}}

\def\sinr{\mathsf{SINR}}
\def\snr{\mathsf{SNR}}

\allowdisplaybreaks 

\begin{document}
\graphicspath{{./Figures/}}
\title{Fundamentals of Throughput Maximization with Random Arrivals for M2M Communications
}
\author{Harpreet S. Dhillon, Howard Huang, Harish Viswanathan and Reinaldo A. Valenzuela
\thanks{
H. S. Dhillon is with WNCG, the University of Texas at Austin, USA. Email: dhillon@utexas.edu. H. C. Huang, H. Viswanathan and R. A. Valenzuela are with Bell Labs, Alcatel-Lucent, NJ. Email: \{howard.huang, harish.viswanathan, reinaldo.valenzuela\}@alcatel-lucent.com. A part of this paper is accepted for presentation at IEEE Globecom 2013 in Atlanta, GA~\cite{DhiHuaC2013}. 
}}


\maketitle
\begin{abstract}
For wireless systems in which randomly arriving devices attempt to transmit a fixed payload to a central receiver, we develop a framework to characterize the system throughput as a function of arrival rate and per-user data rate. The framework considers both coordinated transmission (where devices are scheduled) and uncoordinated transmission (where devices communicate on a random access channel and a provision is made for retransmissions). Our main contribution is a novel characterization of the optimal throughput for the case of uncoordinated transmission and a strategy for achieving this throughput that relies on overlapping transmissions and joint decoding. Simulations for a noise-limited cellular network show that the optimal strategy provides a factor of four improvement in throughput compared to slotted aloha. We apply our framework to evaluate more general system-level designs that account for overhead signaling. We demonstrate that, for small payload sizes relevant for machine-to-machine (M2M) communications (200 bits or less), a one-stage strategy, where identity and data are transmitted optimally over the random access channel, can support at least twice the number of devices compared to a conventional strategy, where identity is established over an initial random-access stage and data transmission is scheduled.

\end{abstract}

\section{Introduction \label{sec:intro}}

Machine-to-machine communications, involving communication between a sensor/actuator and a corresponding application server in the network, are expected to be a major part of cellular networks in the near future~\cite{FocM2003,VodM2010,EriM2011,LieCheJ2011}. While there are millions of M2M cellular devices already using second, third and fourth generation cellular networks, the industry expectation is that the number of devices will increase ten-fold in the coming years~\cite{GSMAM2012,EriM2011a}. Additionally, as has been noted previously in~\cite{ShaJiC2012,DhiHuaJ2013}, the M2M traffic is distinct from consumer traffic, which has
been the main driver for the design of fourth generation communication systems, such as LTE. While current consumer traffic is characterized by small number of long lived sessions, M2M traffic involves a large number of short-lived sessions, typically involving transactions of a few hundred bytes~\cite{JouAttC2011,IEEEM2008}. Because of these differences, there is a strong motivation to optimize cellular networks specifically for M2M communications~\cite{3GPPM2010,ZheHuJ2012,GotLioJ2013}.

\subsection{Related Work}
The problem of enabling M2M communication in cellular networks can be viewed from various related prisms depending upon the metric and set of parameters chosen to be optimized, which are briefly discussed below.

First, the short payloads involved in M2M communications make it highly inefficient to establish dedicated bearers for data transmission. Therefore, in some cases it is better to transmit small payloads in the random access request itself~\cite{CheWanC2010}. More formal treatment of this problem in terms of power and energy minimization is done in~\cite{DhiHuaJ2013,DhiHuaC2012}, where maximum load that a base station can serve is characterized for different access strategies. Several modifications in the current communication protocols to reduce signaling overhead~\cite{MarKirC2005,CheYanC2009,CheWanC2010} and power consumption~\cite{ChaCheC2011} have been proposed in the literature. The problem can also be posed as uplink scheduling problem as is done in~\cite{LioAleC2011}, where the knowledge of the exact delay constraint of all the devices is shown to increase the maximum load that can be served at a base station, compared to the case where devices are partitioned into a limited number of classes. 

Second, a significant number of battery powered devices are expected to be deployed at adverse locations such as basements and tunnels, e.g., underground water monitors and traffic sensors, that demand superior link budgets. 
Motivated by this need for increasing link budget for M2M devices, transmission techniques that minimize the transmit power for short burst communication were studied in~\cite{DhiHuaJ2013}. A related problem of coverage and capacity of M2M in specific realm of LTE is studied in~\cite{RatTanC2012}. Third, the increasing number of M2M devices has led to some new ideas on massive access management~\cite{LoLawC2011,HoHuaJ2012,BarHerC2011,TuHoC2011}. One recurrent theme is that of the cooperative design where the devices are appropriately clustered with one device in each cluster responsible for transmitting all the data generated by that cluster to the base station~\cite{HoHuaJ2012,BarHerC2011,TuHoC2011}.

The need to optimize cellular networks for M2M has also been acknowledged by the standards bodies, e.g., see~\cite{3GPPM2011,3GPPM2011a} for the ongoing efforts in 3GPP. However, despite these research efforts, we still lack a systematic framework to understand the fundamental limits of this new communication regime. A key exception is~\cite{DhiHuaJ2013}, which provided insights into energy and power optimal system design for this regime using tools from optimization. Extending this understanding further, we focus on the throughput optimal design as a function of key system parameters. We consider a realistic model for M2M communications in which a Poisson point process governs the requests for new transmissions, with each request involving sending a fixed number of bits to the base station within a fixed maximum delay. The new tools and insights developed in this paper are summarized next.

\subsection{Contributions}

{\em Maximum throughput:} One can envision multiple modes for sending packets in the uplink, which can be categorized into two broad classes: i) {\em uncoordinated}, where packets are sent in a random access fashion, and ii) {\em coordinated}, where devices transmit on contention-free resources assigned by the base station. We develop a comprehensive framework to determine the maximum throughput that can be achieved for a given arrival rate and outage probability for both these classes under a variety of multiple access strategies. Although our framework is more general, for conciseness we focus on the respective optimal strategies for the two classes and the coordinated and uncoordinated frequency division multiple access (FDMA) strategies. The framework involves formulating an optimization problem involving tradeoffs in outage, rate achievable for each device, and the arrival rate at the base station. Since retransmissions are critical for uncoordinated transmissions, we develop a novel analytic approach to explicitly incorporate them in our framework subject to the maximum delay constraint.

{\em Fundamental result for uncoordinated transmission:} There has been relatively little information theory-based analysis that accounts for the noise and interference caused by simultaneous transmissions \cite{EphHajJ1998}. One recent exception is \cite{MinFraJ2012}, which provides a rigorous information theoretic framework and characterizes the achievable rate region within a guaranteed gap of the optimal region. In this paper, we incorporate in a novel way the joint decoding techniques of coordinated multiple access channel to random access and obtain a new fundamental result characterizing the throughput performance of optimal uncoordinated random access transmission using joint decoding. Using information-theoretic techniques, we obtain an outer bound to the performance and then describe a technique to show that the bound is achievable. Our proposal differs from \cite{MinFraJ2012} in two important ways. First, we assume all users transmit with a fixed rate (in bps/Hz) as derived from the fixed payload size (in bits), the time slot duration (in seconds) and bandwidth (in Hz). Second, we incorporate random arrivals, where the number of users with data to transmit in each time slot are characterized by a Poisson process with a given mean. 

{\em One-stage vs. two-stage designs:} 
Using the proposed framework, we evaluate two classes of system-level designs while accounting for the signaling overhead: i) a {\em one-stage} design, where the data payload is communicated over a random access channel, and ii) a conventional {\em two-stage} design, where the identity of devices is established over a random access channel and the payload is communicated during a scheduled stage. Our analysis demonstrates that one-stage design is more efficient when the information payload is small. However, when the payload is large, packet collisions render one-stage design inefficient and thus two-stage design is superior. We characterize the cross-over payload size for a fairly general set of assumptions on the overhead signaling.

\section{System Model \label{sec:sysmod}}

\begin{table}
\centering
\caption{Notation Summary}
\label{table:notationtable}
\begin{tabulary}{\columnwidth}{ |c | C | }
\hline
    \textbf{Notation} & \textbf{Description} \\ \hline
    $W$			& 	Total bandwidth in Hz \\ \hline
    $\ntaus$		& 	Slot duration in secs\\ \hline
    $\lambda$	& 	Rate of new arrivals at the base station \\ \hline
    $x$			& 	Rate of arrivals including retransmissions \\ \hline
    $g_k$		&	The effective channel gain of $k^{th}$ device	\\ \hline
    $P_{\max}$		&	Maximum power constraint	\\ \hline
    $\mu$		&  	Reference $\snr$ \\ \hline
    $M$			&	Number of minislots per slot \\ \hline
    $\ntaum = \ntaus/M$ &	Minislot duration \\ \hline
    $Z$			& 	Total number of transmissions allowed, including both the first transmission attempt and the subsequent retransmissions \\ \hline
    $T_w$		& 	Length of the retransmission window in minislots \\ \hline
    $\nKslot; \nKsslot$ & New arrivals in each slot; the number of these arrivals that eventually succeed \\ \hline
    $\nKmslot; \nKsmslot$ & New arrivals in each minislot; the number of these arrivals that eventually succeed \\ \hline
    $\ntKmslot; \ntKsmslot$ & Total number of arrivals in a minislot; the number of these arrivals that succeed in that minislot \\ \hline
    $\ncalR; \ncalS$		&  	Maximum common rate; maximum throughput \\ \hline
    $\epsilon$; $\delta$ & The one-shot failure probability in a particular slot or minislot; the eventual failure probability after retransmissions ($\epsilon = \delta$ for no retransmissions)\\ \hline
    $\Theta$ 			& Transmission probability in the optimal uncoordinated strategy \\ \hline
    $\nbbZ_{k_1}^{k_2}$ 	& $\{k_1, k_1 + 1, \ldots, k_2\} \subseteq \nbbZ$ for $k_1 \leq k_2$\\ \hline
    $\{x_k\}_{k_1}^{k_2}$ & $\{x_{k_1}, x_{k_1+1}, \ldots x_{k_2}\}$ for $k_1 \leq k_2$\\ \hline
\end{tabulary}

\end{table}

Consider the uplink of a single cell system in which base station is located at the origin and the devices are uniformly distributed around it in a circle of radius $r_o$. The out-of-cell interference is ignored, which is one of the simulation scenarios in 3GPP model~\cite{3GPPM2010}. The devices transmit such that the exogenous arrivals at the base station can be modeled as a Poisson process with rate $\lambda$ arrivals per second. The 3GPP model also makes similar assumptions about incoming M2M traffic in several simulation scenarios, e.g., see Table 2 in~\cite{3GPPM2012b}. Time is divided into slots of duration $\ntaus$ secs and total available bandwidth is $W$ Hz. Each such slot acts as a ``resource slice'' as shown in Fig.~\ref{fig:resourceblock}. The number of new users with data to transmit in any given slot is denoted by $\nKslot$, which is random and varies according to the Poisson distribution with mean $\lambda \ntaus$. For multiple access, we consider both uncoordinated and coordinated transmission strategies. In uncoordinated access, the base station does not play a role in scheduling users, i.e., it neither decides the set of transmitting users nor their exact scheduling over time-frequency resources. On the other hand, in coordinated access the base station allocates resources for each user's transmission. It involves signaling from the base station to each user indicating the set of transmission resources to be used. Clearly, the uncoordinated access does not require any such signaling.

We assume that each user enters the system with a deadline of $\ntaus$ secs, i.e., it transmits over a single time slot and leaves the system upon successful completion of the data transfer. If the transmission is not successful, the packet is dropped and is said to be in outage. In case of uncoordinated transmission, a slightly more elaborate setup is considered where a slot is divided into $M$ minislots of duration $\ntaum = \ntaus/M$ allowing retransmissions with a limit of $Z$ transmission attempts per packet within $M$ minislots. This implicitly enforces a deadline of $\ntaus$ secs and facilitates fair comparison with the coordinated transmission. Since retransmissions are strictly suboptimal for coordinated transmission, the details of retransmissions are intentionally delayed until Section~\ref{sec:uncoordinated} where we study them in the context of uncoordinated transmission. In both the coordinated and uncoordinated transmission, we restrict attention to transmission strategies in which all users transmit at the same average rate $\ncalR$ over the duration $\ntaus$ whenever they transmit. While this restriction is primarily motivated by the need for simplicity in the case of uncoordinated transmissions, we also impose this restriction on the coordinated access strategies to facilitate fair comparison between the two. Although it may be advantageous to have different users transmitting at different rates, such strategies are outside the scope of this paper.

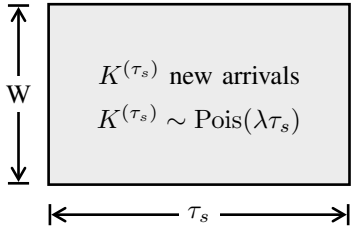
\begin{figure}
\centering
\begin{tikzpicture}[scale=.8,>=angle 60, thick] 
\pgfsetlinewidth{1}	
\coordinate (A) at (0,0);
\coordinate (B) at (5,0);
\coordinate (C) at (5,3);
\coordinate (D) at (0,3);
\coordinate (E) at (2.5,1.5);
\draw [fill=gray!15!,line width = 1] (A) -- (B) -- (C) -- (D) -- cycle;
\node [above] at (E) {$\nKslot$ new arrivals};
\node [below] at (E) {$\nKslot \sim {\rm Pois} (\lambda \ntaus)$};
\draw [|<->|] ($(A)-(.5,0)$) -- ($(D)-(.5,0)$) node [midway,fill=white] {W};
\draw [|<->|] ($(A)-(0,.5)$) -- ($(B)-(0,.5)$) node [midway,fill=white] {$\ntaus$};
\end{tikzpicture}
\caption{A typical time-frequency resource ``slice'' over which $\nKslot$ new arrivals occur.}
\label{fig:resourceblock}
\end{figure}

Each user has a maximum transmission power $P_{\max}$ that is the same for all users. Users may transmit at full power or employ some form of power control and transmit at power levels below the maximum power. For uncoordinated transmission, we assume that the channel state information is not available at the device and hence its transmit power is independent of the channel gain. The uplink channel is modeled as a combination of the power-law path loss, large scale shadowing and small scale fading effects. Therefore, the received power at the base station from a device located at distance $r$ is
\begin{align}
P_r = P_t \ncalX h G r^{-\gamma},
\end{align}
where $P_t$ is the transmit power, $\ncalX$ is a log-normal random variable modeling shadowing gain with standard deviation $\sigma$ dB, $h \sim \exp(1)$ models small scale channel gain due to Rayleigh fading, $G$ is the direction based antenna gain, and $\gamma$ is the path loss exponent. To study the composite effect of all these link budget parameters, we define the reference signal-to-noise ratio ($\snr$) $\mu$ as follows.

\begin{definition}[Reference $\snr$]
The reference $\snr$ $\mu$ is defined as the average received $\snr$ from a device transmitting at maximum power $P_{\max}$ over bandwidth $W$ located at cell edge, i.e., at distance $r_o$ from the base station.
\end{definition}

Therefore, the received $\snr$ $\mu_r$ at the base station from a device located at distance $r \leq r_o$ and transmitting over a bandwidth $W_{\ncalN} \leq W$ can be expressed in terms of $r$ as
\begin{align}
\mu_r = \frac{W}{W_\calN} \frac{P_t}{P_{\max}} \mu \ncalX h \left(\frac{r}{r_o}\right)^{-\gamma},
\end{align}
where the factor $\frac{W}{W_\calN}$ accounts for the difference in the noise power due to the difference in the bandwidths for which $\mu_r$ and $\mu$ are defined. This will be helpful in analyzing the multiple access strategies that involve partitioning of frequency resources, e.g., FDMA. Now defining the effective channel gain as $g = \ncalX h \left(\frac{r}{r_o}\right)^{-\gamma}$, the expression for received $\snr$ $\mu_r$ can be further simplified to
\begin{align}
\mu_r = \frac{W}{W_\calN} \frac{P_t}{P_{\max}} \mu g.
\label{eq:SNR}
\end{align}

For our general discussion throughout the paper, we will assume capacity achieving codes. Please note that the effect of finite block length can be easily incorporated by means of an $\snr$ gap. Interested readers can refer to~\cite{SleJ1963} for more details. 
Under this assumption, for a user transmitting over time $\tau_\ncalN \leq \ntaus$, bandwidth $W_\ncalN \leq W$, and with slight overloading of $\mu_r$ to denote received signal-to-interference-plus-noise ratio ($\sinr$), the rate achieved by the user over a particular time-frequency resource slice can be expressed as
\begin{align}
\ncalR = \frac{\tau_{\ncalN}}{\ntaus} \frac{W_{\ncalN}}{W} \log_2 (1+\mu_r)\ \text{bps/Hz}.
\label{eq:R_def}
\end{align}
The received $\sinr$ depends upon the decoding strategy as discussed in detail in the next two sections. For the special case of FDMA, where a user with channel gain $g$ transmits at $P_{\max}$ over $\frac{W}{B}$ Hz, the rate achieved is
\begin{align}
\ncalR = \frac{1}{B} \log_2\left(1 + \mu_r \right) \stackrel{(a)}{=} \frac{1}{B} \log_2\left(1 + B\mu g \right),
\label{eq:R_FDMAdef}
\end{align}
where the pre-log factor of $\frac{1}{B}$ comes from \eqref{eq:R_def} and $(a)$ follows from \eqref{eq:SNR}. For ease of notation, we denote the set of integers from $k_1$ to $k_2 > k_1$ by $\nbbZ_{k_1}^{k_2} = \{k_1, k_1 + 1, \ldots, k_2\} \subseteq \nbbZ$. Similarly, any general sequence $\{x_{k_1}, x_{k_1+1}, \ldots x_{k_2}\}$ for $k_1 \leq k_2$ is denoted by $\{x_k\}_{k_1}^{k_2}$. The notation used in this paper is summarized in Table~\ref{table:notationtable} for quick reference.

\section{Coordinated Multiple Access \label{sec:coordinated}}
This is the first main technical section of the paper, where we introduce a formal framework to study coordinated multiple access strategies with the goal of maximizing system throughput. Although the framework is general and can be used to study any coordinated transmission strategy, for concreteness we focus on two particular strategies: i) optimal multiuser decoding, which provides a benchmark to evaluate the performance of all other strategies considered in this paper, ii) FDMA, which was shown to be superior than other coordinated multiple access strategies in certain important ways, such as power and energy minimization, in~\cite{DhiHuaJ2013}. Besides, we show that further splitting of time slot into smaller minislots, such as in time division multiple access (TDMA) strategy or hybrid FDMA-TDMA strategy, is strictly suboptimal in the context of throughput maximization under coordinated transmission.

\subsection{Problem Formulation}
We begin this discussion by introducing the main metric of interest for this work, which is the average throughput $\ncalS$, for a given exogenous arrival rate $\lambda$ and a given transmission strategy that captures the average number of successfully transmitted bits per unit time per unit bandwidth. For a slot with $\nKslot$ exogenous arrivals, out of which $\nKsslot$ succeed, the average throughput $\ncalS$ can be expressed as
\begin{align}
\ncalS(\lambda, \mu, \ncalR) = \frac{1}{\ntaus} \nbbE\left[\nKsslot\right] \ncalR,
\label{eq:S_def_cor}
\end{align}
where $\ncalR$ denotes the common rate of each device for the given transmission strategy, $\mu$ is the reference $\snr$ and $\nKslot \sim {\rm Pois}(\lambda \ntaus)$ by assumption. The rest of the arrivals are dropped and the corresponding devices are said to be in outage. Due to the packet deadline of $\ntaus$, these dropped packets cannot be considered for a future transmission. Now assuming $\ncalT_i^{(s)}$ denotes the event that the $i^{th}$ packet, for $i \in \nbbZ_{0}^\nKslot$, succeeds, the number of successful transmissions $\nKsslot$ can be expressed in terms of $\nKslot$ as
\begin{align}
\nKsslot = \sum_{i=1}^\nKslot \nb1 \left( \ncalT_i^{(s)} \right),
\label{eq:nKsslot7}
\end{align}
where $\nb1(\ncalE) = 1$ when event $\ncalE$ occurs and $0$ otherwise. Using \eqref{eq:nKsslot7}, $\nbbE[\nKsslot]$ can be derived as follows
\begin{align}
\nbbE[\nKsslot] &= \nbbE_{\nKslot} \nbbE \left[\sum_{i=1}^\nKslot \nb1 \left( \ncalT_i^{(s)} \right) \big| \nKslot\right] \nonumber\\
&\stackrel{(a)}{=} \nbbE_{\nKslot} \left[\nKslot \nbbP\left(\ncalT_i^{(s)} \right) \right]
\stackrel{(b)}{=} (1-\epsilon) \nbbE[\nKslot],
\label{eq:meanKs_epsilon}
\end{align}
where $(a)$ follows from the linearity of inner expectation, and $(b)$ follows from the outage probability of a packet denoted by $\epsilon$. Thus, the outage probability $\epsilon$ can be expressed in terms of $\nKslot$ and $\nKsslot$ as
\begin{align}
\epsilon = 1 - \frac{\nbbE\left[\nKsslot\right]}{\nbbE\left[\nKslot\right]},
\label{eq:epsilon_c}
\end{align}
using which the average throughput $\ncalS$ can be expressed as
\begin{align}
\ncalS(\lambda, \mu,\ncalR) = \frac{\nbbE[\nKslot]}{\ntaus} \ncalR (1-\epsilon) = \lambda \ncalR (1-\epsilon).
\label{eq:S_c_x}
\end{align}
Given a maximum outage constraint $\epsilon_{\max}$, the throughput maximization problem can now be formulated as
\begin{align}
\begin{array}{cc}
\max\limits_{\ncalR} & \lambda \ncalR (1-\epsilon) \\
{\rm s.t.} & \epsilon \leq \epsilon_{\rm max}
\end{array}.
\label{eq:optimization_c}
\end{align}
We now remark on the solution of this optimization problem.

\begin{remark}[Solution procedure] \label{rem:numprocedure_c}
In addition to the arrival rate and the choice of multiple access strategy, the outage probability $\epsilon$ is also a function of the common rate of the devices. Therefore, for a given arrival rate $\lambda$, there is a maximum common rate $\ncalR$ corresponding to each value of outage probability $\epsilon$. As discussed for optimal and FDMA strategies in this section, it is possible to characterize this relationship analytically. However, the form of this relationship is, in general, such that it does not lead to a closed form analytical result for the maximum throughput. Therefore, we have to resort to the numerical solution for the above optimization problem, which is straightforward once the relationship between the maximum common rate $\ncalR$ and outage probability $\epsilon$ is established. Characterizing this relationship for the optimal and FDMA strategies is the goal of the rest of this subsection.
\end{remark}

To highlight the fact that the outage probability is a function of transmission strategy $\Pi$, arrival rate $\lambda$, reference $\snr$ $\mu$ and common rate $\ncalR$, we let $\ncalE^\Pi(\lambda,\mu,\ncalR)$ denote the outage function as defined by \eqref{eq:epsilon_c}. As remarked above, for a desired outage level $\epsilon$, the maximum common rate $\ncalR^*(\lambda,\epsilon,\mu)$ can be determined as follows
\begin{align}
\ncalR^*(\lambda,\epsilon,\mu) = \arg \max_\ncalR \left[  \ncalE^{\Pi}(\lambda,\mu,\ncalR) \leq \epsilon \right],
\end{align}
for which we need to characterize outage function $\ncalE^\Pi(\lambda,\mu,\ncalR)$, which is done next.

\subsection{Optimal Coordinated Multiple Access}
We first determine the maximum common rate that can be achieved using optimal coordinated transmission by $K$ users with ordered channel gains $\{g_k\}_1^K$ with $g_m \leq g_n$ for $m \leq n$ and reference $\snr$ $\mu$ in a typical resource slice as defined in Section~\ref{sec:sysmod}. The result is given in the following Lemma along with a concise proof.

\begin{lemma}[$\ncalR$ for optimal strategy] \label{lem:R_opt}
The maximum common rate that can be achieved by $K$ users in a typical resource slice with ordered channel gains $\{g_k\}_1^K$ and reference $\snr$ $\mu$ is
\begin{align}
{\widetilde \ncalR}_o \left(\{g_k\}_1^K,\mu\right) = \min_{j\in \nbbZ_{1}^K} \frac{1}{j} \log_2 \left(1+ \mu \sum_{k=1}^j g_k  \right).
\end{align}
\end{lemma}
\begin{proof}
The MAC capacity region for the $K$ users with the given channel gains $\{g_k\}_{1}^K$ and reference $\snr$ $\mu$ consists of $K$-dimensional rate vectors such that
\begin{align}
&{\bf R}^{(\MAC)}\left(\{g_k\}_1^K,\mu \right)= \nonumber \\
& \left\{ \ncalR_k \geq 0: \sum_{k \in X} \ncalR_k \leq \log_2
\left(1+ \mu \sum_{k \in X} g_k \right),  X \subseteq \nbbZ_1^K \right\}.
\end{align}
Within this capacity region, the maximum common rate can be found using a brute-force search over all subsets of users $X \subseteq \nbbZ_1^K$
\begin{align}
{\widetilde \ncalR}_o \left(\{g_k\}_1^K,\mu\right)
&= \min_{X \subseteq \nbbZ_1^K} \frac{1}{|X|} \log_2\left(1+ \mu\sum_{k \in X} g_k  \right) \nonumber\\
&\stackrel{(a)}{=} \min_{j\in \nbbZ_1^K} \frac{1}{j} \log_2 \left(1+ \mu
\sum_{k=1}^j g_k  \right),
\end{align}
where $(a)$ follows from the fact that, for a given $j$, the subset $X$ which minimizes the $\snr$ summation consists of the $j$ users with the smallest channel gains.
\end{proof}

The outage function for the optimal coordinated multiple access strategy, termed as $\Pi_{CO}$, can be expressed in terms of the maximum rate $\widetilde{\ncalR}_o$ derived in the above lemma as follows
\begin{align}
&\ncalE^{\Pi_{CO}}(\lambda,\mu,\ncalR) = \min_{\Pi_{CO}}\left[1 - \frac{{\nbbE}_{\nKslot,\left\{g_k\right\}}\left(\nKsslot\right)}{{\nbbE}(\nKslot)} \right] \nonumber \\
&= 1 - \frac{\max_{\Pi_{CO}} {\nbbE}_{\nKslot,\left\{g_k\right\}}\left(\nKsslot\right)}{{\nbbE}(\nKslot)} \nonumber \\
&\stackrel{(a)}{=} 1 - \frac{ {\nbbE}_{\nKslot,\left\{g_k\right\}}\left(\max_{\Pi_{CO}} \nKsslot\right)}{{\nbbE}(\nKslot)} \nonumber \\
&= 1-\frac{ {\mathbb E}_{\nKslot,\left\{g_k \right\}} \left\{ \arg \max\limits_{K}
\left[{\widetilde \ncalR}_o(\{g\}_{\nKslot-K+1}^\nKslot,\mu) \geq \ncalR \right] \right\}}{\lambda \ntaus},
\label{eq:optoutage_c}
\end{align}
where $(a)$ follows from the fact that under the set of coordinated strategies, maximization over the number of successful transmissions $\nKsslot$ will be performed over each realization. With this characterization of the outage function, it is now possible to numerically evaluate the maximum throughput as defined in \eqref{eq:optimization_c} over the space of coordinated multiple access strategies. The results are presented in Section~\ref{sec:numresults}. We now present a similar formulation of the outage function for FDMA with equal bandwidth allocation next.

\subsection{FDMA with Equal Allocation}

As discussed above for the optimal strategy, we first determine the maximum common rate that can be achieved by $K$ users with ordered channel gains $\{g_k\}_1^K$ with $g_m \leq g_n$ for $m \leq n$ and reference $\snr$ $\mu$ assuming: i) bandwidth is partitioned into $B$ equal subbands, and ii) at most one user is scheduled on a given subband. The maximum common rate in this case is by definition the rate achievable by the user with the lowest channel gain $g_1$, as given in the following Lemma.

\begin{lemma}[$\ncalR$ for FDMA with equal allocation] \label{lem:R_fdma}
The maximum common rate achievable by $K$ users in a typical resource slice with ordered channel gains $\{g_k\}_1^K$ and reference $\snr$ $\mu$ is
\begin{align}
{\widetilde \ncalR}_f(\{g_k\}_1^K,\mu,B) &= \frac{1}{B} \log_2(1+B\mu g_1).
\end{align}
\end{lemma}
\begin{proof}
The maximum common rate is the rate achieved by the weakest user, i.e.,
\begin{align}
{\widetilde \ncalR}_f(\{g_k\}_1^K,\mu,B) &= \min_{g_i} \frac{1}{B} \log_2(1+B\mu g_i),
\label{eq:Rf_intermed}
\end{align}
which follows from~\eqref{eq:R_FDMAdef}. The result now follows from the ordering of channel gains.
\end{proof}

Using this result, the outage function for FDMA with equal allocation, denoted by $\Pi_{CF}$, can be derived as follows
\begin{align}
&\ncalE_R^{\Pi_{CF}}(\lambda,\mu,\ncalR) = \min_{\Pi_{CF}}\left[1 - \frac{\nbbE_{\nKslot,\left\{g_k\right\}}\left(\nKsslot\right)}{{\nbbE}(\nKslot)} \right] \nonumber \\
&= 1 - \frac{\max_{\Pi_{CF}}\nbbE_{\nKslot,\left\{g_k\right\}}\left(\nKsslot\right)}{{\nbbE}(\nKslot)} \nonumber \\
&\stackrel{(a)}{=} 1 - \frac{\nbbE_{\nKslot,\left\{g_k\right\}}\left(\max_{\Pi_{CF}} \nKsslot\right)}{{\nbbE}(\nKslot)} \nonumber \\
&= 1-\frac{ {\mathbb E}_{\nKslot,\left\{g_k \right\}} \left\{\max\limits_B
\left[{\widetilde \ncalR}_f(\{g\}_{\nKslot-B+1}^\nKslot,\mu,B) \geq \ncalR \right] \right\}}{\lambda \ntaus} \nonumber \\
&\stackrel{(b)}{=} 1-\frac{ {\mathbb E}_{\nKslot,\left\{g_k \right\}} \left\{\max\limits_B
\left[g_{\nKslot-B+1} \geq \frac{2^{B\ncalR}-1}{B\mu} \right] \right\}}{\lambda \ntaus},
\label{eq:fdmaoutage_c}
\end{align}
where $(a)$ follows from the fact that in coordinated transmission, maximization over $\Pi_{CF}$ is performed in each realization, and $(b)$ follows from \eqref{eq:Rf_intermed}. The maximum throughput is now derived numerically using this characterization of the outage function and the results are presented in Section~\ref{sec:numresults}. Before concluding this section, we show that splitting a time slot into smaller minislots as a part of TDMA or hybrid TDMA-FDMA strategies leads to a lower common rate than FDMA, which is the reason why it is not considered in the discussion above.  For fair comparison with FDMA with equal bandwidth allocation, assume that the time slot is divided into $M  \geq 1$ minislots of equal duration and the bandwidth is divided into $B \geq 1$ equal frequency bins, leading to $MB$ time-frequency resource blocks each scheduling at most one user. Note that for $M=1$, this reduces to FDMA with equal bandwidth allocation. Now assume that $\nKsslot = MB$ users, with channel  gains $\{g_k\}_{1}^{\nKsslot}$, are scheduled over $MB$ resource blocks. As discussed in Lemma~\ref{lem:R_fdma}, the maximum common rate, corresponds to the user with weakest channel gain $g_1$. Therefore,
\begin{align}
{\widetilde \ncalR}_{tf} &= \frac{1}{MB} \log_2 (1+B\mu g_1)\nonumber \\
&= \frac{1}{\nKsslot} \log_2 \left(1+\frac{\nKsslot}{M}\mu g_1\right),
\end{align}
which is maximized for $M=1$, thereby showing that splitting a time slot into minislots is not optimal for coordinated transmission. However, this is not the case in uncoordinated transmission, as discussed in detail in the following section.

\section{Uncoordinated Multiple Access \label{sec:uncoordinated}}

This is the second main technical section of this paper where we study the maximum throughput for uncoordinated strategies, or equivalently, transmission strategies over a time-slotted random access channel. Unlike the coordinated transmission discussed in the previous section, it is important to consider retransmissions here. Therefore, before discussing the main problem formulation, we first incorporate retransmissions in the system model with special focus on characterizing the effective arrival rate and deriving the effective failure probability after retransmissions. This forms the first main technical contribution of this section. We propose a novel multiuser detection strategy and establish its optimality for uncoordinated transmission, which forms the second main technical contribution of this section. For fair comparison with the coordinated strategies studied in the previous section, we also consider uncoordinated FDMA with equal bandwidth allocation in which the total bandwidth is partitioned into subbands of equal bandwidth and a transmitter chooses to transmit on a randomly selected subband using the slotted aloha protocol. For the same reason as discussed for coordinated transmission, TDMA and hybrid TDMA-FDMA strategies lead to lower common rate than FDMA, and are hence not considered in this discussion. Another popular random access strategy is code-division-multiple-access (CDMA), which was recently shown to perform better than random access FDMA for transmit power and energy minimization when the channel gains are known at the transmitters~\cite{DhiHuaJ2013}. However, for throughput maximization under no channel state information at the transmitter (CSIT), random access FDMA is known to perform better than random access CDMA~\cite{VisJ2009} because of which we do not consider CDMA in this study. We now introduce the formal setup along with the details of the slot structure and discuss retransmissions in detail in the following subsection.
\begin{figure}[t]
\centering
\includegraphics[width=1\columnwidth]{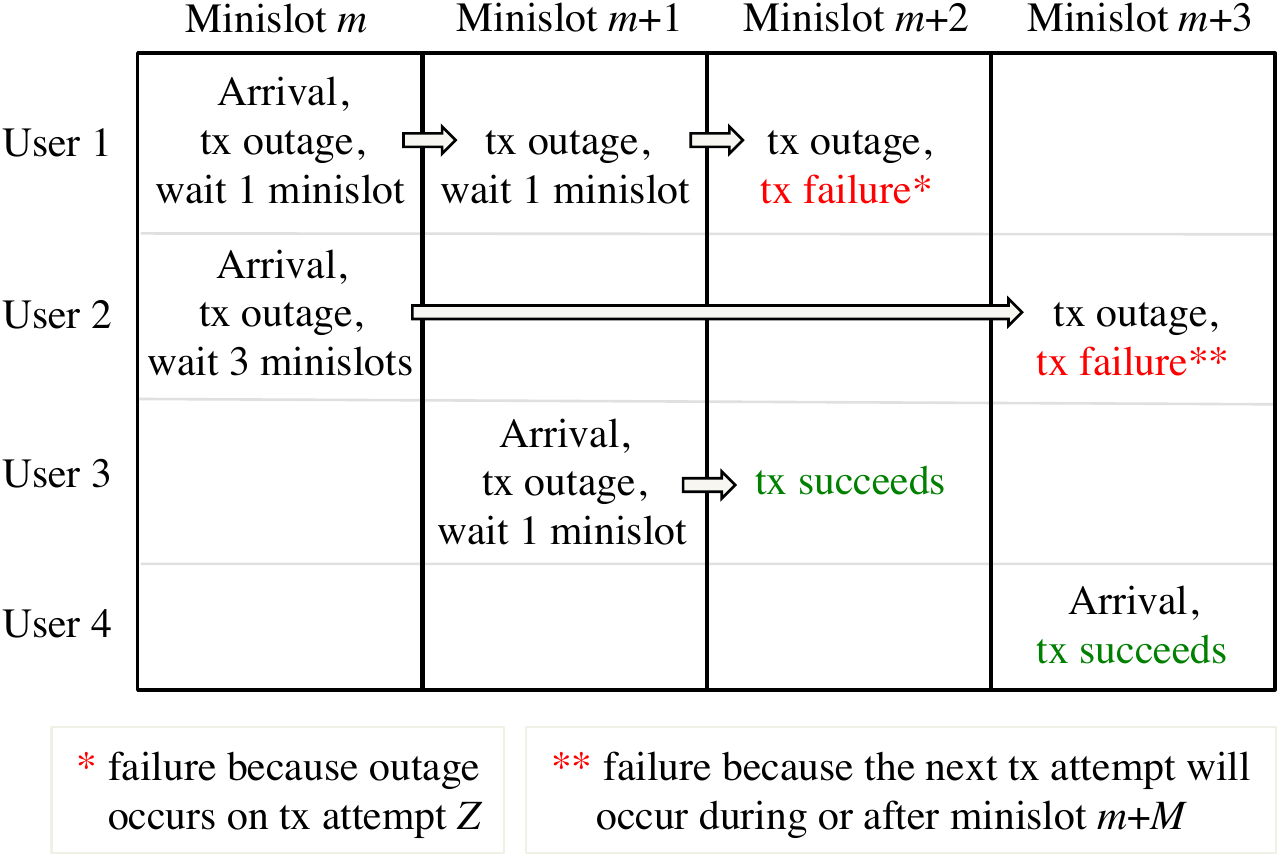}
\caption{An illustration of the retransmission process. The two events leading to packet failures are highlighted for $Z=3$ and $M=4$. See example~\ref{eg:ReTxExample}.}
\label{fig:ReTxExample}
\end{figure}

\subsection{Modeling Retransmissions}
We assume that each slot of duration $\ntaus$ is divided in $M$ minislots with equal duration $\ntaum = \ntaus/M$ as shown in Fig.~\ref{fig:ReTxExample}, where a slot is divided into $4$ minislots. Other details appearing in Fig.~\ref{fig:ReTxExample} will be discussed in Example~\ref{eg:ReTxExample}. As in the previous section, assume that each user comes with a delay constraint of $\ntaus$ secs, i.e., it cannot remain in the system for more than $M \geq 1$ minislots. Therefore, in the rest of this section, $M$ will be used both to denote the number of minislots and the deadline of each user. Further assume that each user is allowed at most $Z$ retransmissions, which also counts the first transmission attempt when the user enters the system for the first time\footnote{The impact of retransmission limit on delay and throughput of preamble contention in the specific setup of LTE-A random access is studied in~\cite{TyaAurJ2012}.}. To facilitate analysis, we make following assumption about the composite arrival process of the new packets and the retransmitted packets.
\begin{assumption}[Composite arrival process] \label{as:arrprocess}
The net arrival process of the new packets and the retransmitted packets in the steady state is approximated by a homogeneous Poisson process with density $x \geq \lambda$. This assumption is reasonably accurate, especially when the back-off times are large and the number of retransmissions is not too large~\cite{AbrC1970}.
\end{assumption}
The back-off time is random and is assumed to be uniformly distributed over a window of $T_w$ minislots, i.e., the retransmission window is $\{1, 2, \ldots, T_w\}$. The uniform distribution for the back-off time is well accepted due to its finite support and reasonably simple implementation~\cite{BiaJ2000,3GPPM2012}. For simplicity, we assume $T_w$ to be a constant, which does not depend upon the retransmission counter. As will be evident from this discussion, more sophisticated strategies, such as the exponential back-off strategy implemented in 802.11 MAC, where $T_w$ is increased exponentially with the retransmission counter, can also be studied with a slight modification of our analysis. In this setup, the packet failure occurs under one of the following two events: i) allowed maximum number of retransmissions are exhausted, or ii) the packet deadline is expired. Denote by $\delta(\lambda, \mu, \ncalR, T_w, M, Z)$ the failure probability, i.e., at least one of the two events listed above occurs. The probability of unsuccessful transmission at a particular attempt is denoted by $\epsilon$, as in the previous section, which is a function of the transmission technique as well as the effective arrival rate after retransmissions. Note that $\delta = \epsilon$ whenever there is only one transmission attempt allowed for each packet, i.e., $Z=M=1$ (and $T_w$ is irrelevant). Note also that the average backoff is $(T_w+1)/2$ minislots, and the average number of transmission attempts is $2M/(T_w+1)$ if we assume each transmission is in outage. To fix these ideas, we consider the following example.

\begin{example}[Retransmissions] \label{eg:ReTxExample}
Consider the simple scenario depicted in Fig.~\ref{fig:ReTxExample}, where $Z=3$ and $M=4$. The minislot index $m$ is with respect to an arbitrary reference. The user indexing is also arbitrary. In minislot $m$, two users arrive in the system. User 1 fails in all the first three minislots, thus exhausting the maximum number of retransmissions and leading to the packet failure. User 2 retransmits after waiting for three minislots but fails. Its deadline is expired, which leads to the packet failure. User 3 arrives in minislot $m+1$. It fails in the first attempt but retransmits and succeeds in the very next minislot. User 4 arrives in minislot $m+3$ and succeeds in the very first attempt. The failure events corresponding to users 1 and 2 contribute towards $\delta$ as discussed next.
\end{example}

Denote by $Y_n$ the number of users with data to transmit in a minislot $n$. The index $n$ counts the number of minislots without differentiating between the slots to which they belong. It is important to note that distinguishing minislots based on the slots from which they are created is not important because of Assumption~\ref{as:arrprocess}. We can now write the following equation required for the equilibrium,
\begin{align}
\nbbE[Y_n] = \lambda\ntaum + \sum_{j< n} \nbbE[Y_{n-j}] p_j \epsilon - \lambda\ntaum \delta,
\end{align}
where $\lambda\ntaum$ is the average number of new arrivals in a minislot and $p_j$ is the probability that the unsuccessful user makes a new transmission attempt after $j$ minislots. In steady state, $\nbbE[Y_n] = x \ntaum$ for all $n$, where $x \geq \lambda$ is the effective arrival rate after retransmissions. Hence the above equation can be rewritten as 
$x = \lambda +  x \epsilon - \lambda \delta$, 
from which the net arrival rate $x$ can be evaluated as
\begin{align}
x = \lambda \frac{1-\delta}{1-\epsilon}.
\label{eq:x_expression}
\end{align}
Since $x \geq \lambda$, it immediately follows that $\delta \leq \epsilon$, i.e., the effective failure probability is always smaller than the per-slot failure probability.

\begin{remark} \label{Rem:Delta_Epsilon}
There are two equivalent ways of formally treating retransmissions in our framework. First is to look at only the exogenous arrivals with the failure probability of each arrival being $\delta$. This is because every new arrival in the system eventually fails with probability $\delta$. Second is to look at the net arrival process, which is a Poisson process with arrival rate $x$, with failure probability of each arrival being $\epsilon$. This equivalence is evident in~\eqref{eq:x_expression} and will be helpful when we formally pose the problem of common rate maximization later in this section.
\end{remark}

We now derive an expression for $\delta$ as a function of $M$, $Z$ and $T_w$. For the derivation, we need the following technical result about the distribution of the sum of discrete uniformly distributed random variables. For the proof of this result, please refer to~\cite{CaiRatC2007}.
\begin{lemma}
\label{lem:UnifSum}
Let $S_n = \sum_{i=1}^n X_i$ be the sum of $n$ discrete i.i.d. uniformly distributed random variables over $\{1, 2, \ldots T_w\}$. The p.m.f. of $S_n$ is given by
\begin{align}
\nbbP[S_n = n+j] = \left(\frac{1}{T_w} \right)^n {n \choose j}_{T_w}, j \in \nbbZ_{0}^{n(T_w-1)},
\end{align}
where the second multiplicative term in the above expression is a polynomial coefficient, i.e., ${n \choose i}_{k+1}$ is the coefficient of $y^i$ in the expansion of $(1+y+\ldots y^k)^n$ $\forall i \in \nbbZ_{0}^{nk}$ and $0$ for all other values of $i$. Another alternate representation of the p.m.f. in terms of Gamma function is
\begin{align}
&\nbbP[S_n = n+j] = \nonumber \\
&\frac{n}{(T_w)^n} \sum_{p=0}^{\lfloor \frac{j}{T_w} \rfloor} \frac
{\Gamma(n+j-pT_w) (-1)^p}
{\Gamma(p+1)\Gamma(n-p+1)\Gamma(j-pT_w+1)}.
\end{align}
\end{lemma}
Using this Lemma, we now derive the failure probability $\delta$. The result is stated in the following theorem and the proof is given in Appendix~\ref{App:Delta}.
\begin{theorem}[Failure probability]  \label{thm:Delta}
The probability that the packet transmission eventually fails due to the deadline expiration or the exhaustion of maximum allowed number of retransmissions is
\begin{align}
\delta &= \epsilon\left(1 - \frac{1}{T_w}\sum_{i=1}^{M - 1} \nb1\left(i \in \nbbZ_1^{T_w}\right) \right) +  \nonumber \\
&\sum_{n=2}^{Z-1} \frac{\epsilon^n}{T_w} \sum_{j=1}^{T_w}
\sum_{k=M - j}^{M - 1} \nb1 \left( k \in \nbbZ_{n-1}^{(n-1)T_w} \right) \frac{1}{T_w^{n-1}} {n-1 \choose k-n}_{T_w} \nonumber \\
&+ \epsilon^Z \sum_{j=0}^{M-Z} \nb1 \left( j \in \nbbZ_{0}^{(Z-1)(T_w-1)}\right) \frac{1}{(T_w)^{Z-1}} {Z-1 \choose j}_{T_w}.
\label{eq:Delta_expression}
\end{align}
\end{theorem}

Recall from the previous section that the outage probability $\epsilon$ depends upon the arrival rate, maximum common rate and the transmission strategy. Under no retransmissions, the arrival rate is simply $\lambda$, which allowed us to formulate outage function $\ncalE^{\pi}(\lambda,\mu,\ncalR)$ for given $\lambda$ and reference $\snr$ $\mu$. However, under retransmissions the net arrival rate $x$ itself is a function of $\epsilon$ as evident from~\eqref{eq:x_expression}. Using the closed form expression derived for $\delta$ in Theorem~\ref{thm:Delta}, we can express $\delta$ in terms of $\epsilon$ in~\eqref{eq:x_expression}. Now expressing $\epsilon$ as the outage function $\ncalE^\Pi(x,\mu,\ncalR)$ to highlight its dependence upon the aggregate arrival rate $x$, \eqref{eq:x_expression} becomes a fixed point equation in terms of $x$. Given the outage function $\ncalE^\Pi(x,\mu,\ncalR)$, it can be iteratively solved to obtain $x$ as a function of $\ncalR$ for a given $\mu$.  Then $\delta$ can be obtained using \eqref{eq:Delta_expression}. The details of this procedure are provided in the following remark along with some comments on the existence and uniqueness of the said fixed point. The outage function for the uncoordinated optimal and FDMA strategies will be characterized later in this section.

\begin{remark}[Procedure to iteratively compute $x$] \label{rem:numx_uc}
We let ${\widetilde x}[i]$ denote the arrival rate $x$ on the $i^{th}$ iteration, where $i\in \nbbZ^+$. It is computed interactively as 
 \begin{equation}
\label{eq:tilde_x_iter}
{\widetilde x}[i+1] =  \Psi({\widetilde x}[i]),
\end{equation}  where, using \eqref{eq:x_expression}, we define
\begin{equation}
\label{eq:g_tilde_x}
\Psi({\widetilde x}) = \lambda \frac{\left(1- \sum_{n=1}^Z a_n \left(\ncalE^\Pi(\widetilde{x},\mu,\ncalR)\right)^n  \right)}{1-\ncalE^\Pi(\widetilde{x},\mu,\ncalR)},
\end{equation}
and the probability $\nbbP[\ncalA_n]$ from the expression of $\delta$ given by~\eqref{eq:Delta_intermed} is denoted by $a_n$ for the ease of notation. Note that by definition $0 \leq a_n \leq 1$ $\forall n \in \nbbZ_1^Z$. By initializing $x[1] = \lambda < \infty$, the desired fixed point is given by $x = \lim_{i\rightarrow \infty}{\widetilde x}[i]$ using (\ref{eq:tilde_x_iter}). The fact that the iterative strategy results in a finite $x$ hinges on the observation that outage $\ncalE^\Pi(x,\mu,\ncalR)$ increases monotonically as the effective arrival rate ${\widetilde x}$ increases. Therefore from \eqref{eq:Delta_expression}, the failure probability $\delta$ also increases monotonically as ${\widetilde x}$ increases because the coefficients $a_n$ are non-negative. Since $\delta({\widetilde x}) \leq \ncalE^\Pi(x,\mu,\ncalR)$ by definition, it follows that $\Psi({\widetilde x})$ in \eqref{eq:g_tilde_x} increases monotonically as ${\widetilde x}$ increases. However, as ${\widetilde x}$ increases without bound, it can be shown using L'Hopital's rule that
\begin{equation}
\lim_{{\widetilde x} \rightarrow \infty} \Psi({\widetilde x}) = \lambda \sum_{n=1}^Z na_n.
\end{equation}
Because $\Psi({\widetilde x})$ increases with ${\widetilde x}$ but is bounded, it follows that the iterations in \eqref{eq:tilde_x_iter} will yield a finite fixed point which is the desired effective arrival rate $x$.
\end{remark}
Note that the iterative procedure explained above to determine net arrival rate $x$ requires outage function $\ncalE^\Pi(x,\mu,\ncalR)$. This function will be characterized both for the uncoordinated optimal and FDMA strategies after formulating the problem of maximizing throughput, which is done next.

\subsection{Maximum Throughput}
As discussed before, the main metric of interest for this work is the average throughput $\ncalS$. Following the equivalence arguments of Remark~\ref{Rem:Delta_Epsilon}, the average throughput can be formulated in two equivalent ways. The first way is to focus only on the new arrivals in each minislot and count the ones that eventually succeed. Since retransmissions are allowed, these arrivals may succeed in a future minislot. For given exogenous arrival rate $\lambda$ and a minislot with $\nKmslot$ new arrivals, out of which $\nKsmslot$ eventually succeed, each with probability $1-\delta$, the average throughput $\ncalS$ can be expressed as
\begin{align}
\ncalS(\lambda, \mu, \ncalR) = \frac{1}{\ntaum} \nbbE[\nKsmslot] \ncalR,
\end{align}
where $\frac{\nbbE[\nKsmslot]}{\ntaum}$ denotes the average number of transmissions that eventually succeed in a unit time. Following the same arguments as in \eqref{eq:S_def_cor}-\eqref{eq:meanKs_epsilon}, we can express $\nbbE[\nKsmslot] $ as a function of $\delta$ and $\nbbE[\nKmslot]$ as
\begin{align}
\nbbE[\nKsmslot] = (1-\delta) \nbbE[\nKmslot],
\end{align}
using which the average throughput can be expressed as
\begin{align}
\ncalS(\lambda, \mu, \ncalR) &= \frac{\nbbE[\nKmslot]}{\ntaum}  \ncalR (1-\delta)
 = \lambda \ncalR (1-\delta).
\end{align}
Using \eqref{eq:x_expression}, it can also be equivalently expressed in terms of the aggregate arrival rate $x$ as
\begin{align}
\ncalS(\lambda, \mu, \ncalR) = x \ncalR (1-\epsilon).
\label{eq:S_uc_x}
\end{align}
Interestingly, the above expression can be directly derived if we focus on the aggregate arrivals in each minislot instead of just the new arrivals. In this case, for a given aggregate arrival rate $x$, total number of arrivals in a given minislot $\ntKmslot$, out of which $\ntKsmslot$ succeed in that minislot, the average throughput can be expressed as
\begin{align}
\ncalS(\lambda, \mu, \ncalR) = \frac{1}{\ntaum} \nbbE\left[\ntKsmslot\right] \ncalR.
\end{align}
As done for the new arrivals above, $\ntKsmslot$ can be expressed in terms of $\ntKmslot$ using the fact that each transmission succeeds in that minislot with probability $1-\epsilon$. The expression is
$\nbbE \left[ \ntKsmslot \right] = \nbbE \left[ \ntKmslot \right] (1-\epsilon)$, 
which directly leads to \eqref{eq:S_uc_x}. As a byproduct, $\epsilon$ can be expressed in terms of $\ntKmslot$ and $\ntKsmslot$ as
\begin{align}
\epsilon = 1 - \frac{\ntKsmslot}{\ntKmslot},
\label{eq:epsilon_uc}
\end{align}
which will be useful in formulating the outage function for the optimal and FDMA strategies later in this section. Now an optimization problem similar to~\eqref{eq:optimization_c} can be formulated to maximize the throughput given an outage constraint $\delta_{\max}$. Recall that for no retransmission case, $\delta=\epsilon$, which is why the outage constraint in \eqref{eq:optimization_c} was given in terms of $\epsilon_{\max}$. The optimization problem is
\begin{align}
\begin{array}{cc}
\max\limits_{\ncalR} & \lambda\ncalR(1-\delta) \\
{\rm s.t.} & \delta \leq \delta_{\rm max}
\end{array}.
\label{eq:optimization_uc}
\end{align}
The general numerical procedure to evaluate maximum throughput is the same as for the coordinated transmissions discussed in Remark~\ref{rem:numprocedure_c} but includes an additional step of evaluating the aggregate arrival rate as described in Remark~\ref{rem:numx_uc}. The complete procedure is briefly summarized below.

\begin{remark}[Solution procedure] \label{rem:numprocedure_uc}
To solve the optimization problem given by \eqref{eq:optimization_uc}, we need to find the failure probability $\delta$ as a function of the common rate $\ncalR$ for a given exogenous arrival rate $\lambda$. This can be achieved in three steps. First step is to determine aggregate arrival rate $x$ for given $\ncalR$, $\lambda$ and a given outage function $\ncalE^{\Pi}$ using the procedure described in Remark~\ref{rem:numx_uc}. Second step is to determine the value of the outage function $\ncalE^{\Pi}$ (say $\epsilon$) at this aggregate arrival rate $x$. For this $\epsilon$, the failure probability $\delta$ can be determined using \eqref{eq:Delta_expression}. The maximum throughput can now be evaluated numerically.
\end{remark}

We now characterize outage function $\ncalE^\Pi(x,\mu,\ncalR)$ for both the optimal and the FDMA strategies.

\subsection{Optimal Uncoordinated Multiple Access}
Using \eqref{eq:epsilon_uc}, the outage function for the optimal uncoordinated strategy can be expressed as
\begin{align}
\ncalE^{\Pi_{UO}}(x,\mu,\ncalR) &= \min_{\Pi_{UO}}\left[1 - \frac{\nbbE_{\ntKmslot,\left\{g_k\right\}}\left(\ntKsmslot\right)}{\nbbE\left(\ntKmslot\right)} \right]   \nonumber \\
&= 1 - \frac{\max_{\Pi_{UO}} \nbbE_{\ntKmslot,\left\{g_k\right\}}\left(\ntKsmslot\right)}{\nbbE\left(\ntKmslot\right)}\nonumber \\
&= 1 - \frac{\max_{\Pi_{UO}} \nbbE_{\ntKmslot,\left\{g_k\right\}}\left(\ntKsmslot\right)}{x\ntaum}.
\label{eq:optoutage_uc}
\end{align}
The goal now is to derive $\nbbE\left[\ntKsmslot\right]$, which involves a novel multiuser decoding strategy as discussed in the next theorem. Please note that the proposed strategy is described in the achievability part of the proof. Also note that even though this theorem is presented in the context of a minislot, the result is general and applies to any given time-frequency resource block with a given arrival rate. For notational simplicity, we denote the maximum common rate per minislot by $\ncalR_m = M \ncalR$ in the following result.

\begin{theorem} \label{thm:EK_uc}
The mean number of users that succeed in a given minislot in the optimal uncoordinated transmission strategy is
\begin{align}
\nbbE\left[\ntKsmslot\right] = \max_{\Theta} x \ntaum \Theta \nbbE_{\Omega} \left[\max_{|\ncalL|}\frac{|\ncalL|}{\Omega} \nbbP \left(\ncalC(P_{\max}) \right)\right],
\end{align}
where $\Omega\sim {\rm Pois}(x \ntaum \Theta)$ and the event $\ncalC(P_{\max})$ is
\begin{align}
\left\{ |{\widetilde \ncalL}|\ncalR_m \leq  \log_2 \left( 1 + \frac{\mu \sum_{i \in {\widetilde \ncalL}} g_{i}}{1+\mu \sum_{m \in \ncalT-\ncalL}g_{m}}\right), \forall {\widetilde \ncalL} \subseteq \ncalL \right\}.
\end{align}
\end{theorem}

The proof of this theorem involves both converse and achievability parts, discussed in detail below.
\begin{proof}
{\bf (Converse)} To prove the converse, consider any strategy $\Pi_U$ for uncoordinated transmission that governs: i) the decision to transmit, and ii) the power level of a user which has data to transmit in any given minislot independent of total number of users $\ntKmslot$ in that minislot and their channel gains $\{g_{k}\}_{1}^\ntKmslot$. For simplicity, we drop the superscript and denote the number of users (aggregate arrivals) by $\ntKmslots$. Similarly, we drop the superscript from $\ntKsmslot$ and denote the number of successful transmissions in a minislot by $\ntKsmslots$. 
Let ${\widetilde s}(n)$ be the indicator for whether user $n$ transmits, i.e., ${\widetilde s}(n)=1$ when user $n$ transmits and ${\widetilde s}(n)=0$ otherwise. Let $\nbbE[{\widetilde s}(n)] = \nbbP({\widetilde s}(n)=1) = \Theta$. Let $P_{n} \leq P_{\max}$ be the transmit power of user $n$ that is transmitting. Note that for any uncoordinated transmission strategy, ${\widetilde s}(n)$ and $P_{n}$ are independent of $\ntKmslots, \{g_{k}\}_{1}^\ntKmslots$ and independent of each other. Denote the set of transmitting users by $\ncalT$, i.e.,
\begin{align} \ncalT = \left\{ n: { \widetilde s }(n)=1 \right\}.
\end{align}
Denote the common rate by $\ncalR$. Since the transmissions are confined over a minislot of duration $\ntaus/M$, the required rate to achieve common rate $\ncalR$ is $\ncalR_m = M\ncalR$. Please refer to~\eqref{eq:R_def} for more details.
Now suppose a subset $\ncalL \subset \ncalT$ of users' messages are successfully decoded by the receiver. Then from the multiple access channel (MAC) channel theorem we require that
\begin{align}
 l \ncalR_m \leq  \log_2 \left( 1 + \sum_{i \in {\widetilde \ncalL}}\frac{ \frac{P_{i}}{P_{\max}}\mu g_{i}}{1+\sum_{m \in \ncalT-\ncalL}\frac{P_{m}}{P_{\max}}\mu g_{m}}\right), \forall {\widetilde \ncalL} \subset \ncalL,
 \end{align}
where $l=|{\widetilde \ncalL}|$. For ease of notation, we denote this event by $\ncalC (\{P_n\},\{g_n\},\mu, \ncalL,\ncalT)$. Recall that the same concept was also used in the proof of Lemma~\ref{lem:R_opt}.
In the rest of the proof, we will use the short hand notation $\ncalC(\{P_n\})$, with the understanding that it does depend upon other parameters but they are not important for this discussion. The average number of users that are successfully decoded can now be upper bounded by
\begin{align}
\nbbE[\ntKsmslots] &\leq  \max_{\{P_{n}\}}{\mathbb E}_{\ntKmslots, {\widetilde s}(n)}
 {\mathbb E}_{\{g_{n}\}}\left[\max_{|\ncalL|} |\ncalL| \nb1 \left( \ncalC (\{P_n\}) \right) \big| \ntKmslots,\{{\widetilde s}(n)\} \right] \nonumber \\
& \stackrel{(a)}{\leq}  \max_{P} \nbbE\left[ \nbbE \left[\max_{|\ncalL|}|\ncalL| \nb1 \left(\ncalC(P) \right) \big| \ntKmslots, \{{\widetilde s}(n)\}\right] \right]\nonumber \\
& \stackrel{(b)}{\leq}  \nbbE \left[ \nbbE \left[\max_{|\ncalL|}|\ncalL| \nb1 \left(\ncalC(P_{\max}) \right) \big| \ntKmslots, \{{\widetilde s}(n)\}\right] \right]\nonumber \\
& \stackrel{(c)}{=}  \nbbE  \nbbE \left[\frac{\sum_{n=1}^{\ntKmslots}{\widetilde s}(n)}{\Omega} \max_{|\ncalL|}|\ncalL| \nb1 \left(\ncalC(P_{\max}) \right) \big| \ntKmslots, \{{\widetilde s}(n)\}\right]\nonumber \\
& \stackrel{(d)}{=} \nbbE \left[\sum_{n=1}^{\ntKmslots} {\widetilde s}(n)  \right] \nbbE_{\Omega, \{g_n\}} \left[\max_{|\ncalL|}\frac{|\ncalL|}{\Omega} \nb1 \left(\ncalC(P_{\max}) \right)\right]\nonumber \\
&= x \ntaum \Theta \nbbE_{\Omega} \left[\max_{|\ncalL|}\frac{|\ncalL|}{\Omega} \nbbE_{\{g_n\}} \nb1 \left(\ncalC(P_{\max}) \right)\right]\nonumber \\
&= x \ntaum \Theta \nbbE_{\Omega} \left[\max_{|\ncalL|}\frac{|\ncalL|}{\Omega} \nbbP \left(\ncalC(P_{\max}) \right)\right],
\label{eq:converse-rateconstraint}
\end{align}
where $(a)$ follows from the fact that the transmit power is independent of all the other variables, including channel gains of the respective user, which means that ``optimal'' value will be the same for all the users, $(b)$ follows from the fact that the function 
$\sum_{i \in \ncalL}\frac{\frac{P}{P_{\max}}\mu g_{i}}{1+\sum_{m \in {\widetilde \ncalT}-\ncalL}\frac{P}{P_{\max}} \mu g_{m}}$  
is an increasing function of $P$ for given set of channel gains, $(c)$ follows by simply setting $\Omega=\sum_{n=1}^{\ntKmslots}{\widetilde s}(n)$, which is a Poisson distributed random variable with mean $x \ntaum \Theta$, $(d)$ follows from the independence of ${\widetilde s}(n)$ from all other random variables. This completes the proof of the converse.

{\bf (Achievability)} We now present the achievability proof, i.e., show that the above upper bound on the average number of successfully decodable users is achievable as the number of channels symbols per slot goes to infinity. First, pick a set of $\bar{N}$ code books for the AWGN MAC channel, each with $2^{\ncalR_m}$ code words of length ${\bar n}$.  The code words can be picked at random from the typical set of Gaussian random input distribution with transmit power $P_{\max}$ as is usually done in the proof of the MAC channel coding theorem \cite{CovThoB1991}. Observe that the code books for each user are selected independently of the other users and the $\snr$. Furthermore, the code book is specifically not dependent on the channel gain except for the number of code words, which is governed by the rate $\ncalR_m$. Thus a combination of such independently chosen code books can form the code book of the MAC channel.

At the beginning of each transmission, a preamble Gaussian sequence of length $q$ that is unique to that code book is transmitted to help the receiver detect which code books are being transmitted. All the code words of the same code book will have the same sequence while the different code books will have distinct sequences. The sequences can be drawn at random from a Gaussian random variable with variance $P_{\max}$.

As before, suppose there are $\ntKmslots$ users who have data to transmit in a given slot. In our strategy, each user will decide to transmit independently with probability $\Theta$. Recall that the total number of transmitting users is denoted by $\Omega$. When a user decides to transmit, it will pick one of the code books from $\bar{N}$ code books at random and map the message to one of the code words in the chosen code book in the usual way and first transmit the preamble followed by the code word.

The collision probability $P_{c}$ that two or more transmitting users pick the same code book when $\Omega$ users attempt to transmit is given by
\begin{align}
P_{c}(\Omega) = 1 - \left(1- \frac{1}{\bar{N}} \right)^{\Omega-1},
\end{align}
which is proved in Lemma 1 of~\cite{DhiHuaJ2013}.
Since $\Omega$ is a Poisson random variable there exists $\Omega_{\max}$ such that $\nbbP\{ \Omega > \Omega_{\max} \} < \delta'/2$. By picking $\bar{N}$ large enough so that $P_{c}(\Omega_{\max}) \leq \delta'/2$, we can ensure that the overall collision probability $P_{c} \leq \delta'$.

Now consider only the case when transmitting users have picked distinct code books. We propose a decoding strategy in which the receiver first detects which code books are in use, i.e, code words from which code books are transmitted in that slot, and then proceed to decode the code words. With a simple correlator detector, the receiver can detect the presence of each of the possible code books. The total detection error, i.e., detecting positively a code book that is not used or missing a code book that is in use, can be bounded by
\begin{align}
P_{d} \leq &\bar{N} \exp\left(-q \Phi \left( \frac{\mu g_{\min}}{1 + (\bar{N}-1)\mu g_{\max}} \right)\right)+ \nonumber \\
&\bar{N} \left(\nbbP\left( g < g_{\min} \right) + \nbbP\left( g > g_{\max} \right) \right),
\end{align}
where we bounded the individual detection error probabilities assuming worst case channel scenario where the interferers have some large channel gain while the desired signal being detected has a small channel gain. Observe that it is possible to first pick $g_{\min}$ and $g_{\max}$ so that $\bar{N} \left( \nbbP\left( g < g_{\min} \right) + \nbbP\left( g > g_{\max} \right) \right) \leq \delta''/2$ and then pick $q$ sufficiently large so that $P_{d} \leq \delta''$. Furthermore, since $q$ is fixed relative to block length, the loss in capacity because of the preamble can be made arbitrarily small. The resulting loss in transmission rate can thus be made negligible. Therefore, we can assume that the receiver knows which code books are in use. Note that since the channels are static the receiver can similarly first estimate the channel accurately from the preamble and then apply the joint typicality test for the known channels. We will assume that the channel is estimated without errors observing that the increase in decoding error probability because of channel estimation errors can be shown to be arbitrarily small when $q$ goes to infinity.

Once the channels are determined the receiver tries to evaluate which subset of code books from the detected code books can actually be decoded treating the rest as noise. To this end, the receiver can compute the rate region for decoding all possible subsets of code books detected and pick the largest set of users for which the rate $\ncalR_m$ is achievable. In other words, the receiver can find the subset $\ncalL$ of the set of code books $\ncalT$ detected to be in use. Thus, the receiver will attempt to decode the maximal set of code words from the code books for which
\begin{eqnarray}
|{\widetilde \ncalL}| \ncalR_m \leq \log_2 \left( 1 + \sum_{i \in {\widetilde \ncalL}}\frac{\mu g_{i}}{1+\sum_{m \in \ncalT-\ncalL}\mu g_{m}} \right), \forall {\widetilde \ncalL} \subset \ncalL. \label{eq:rate constraint}
 \end{eqnarray}

From the MAC channel theorem, we know that whenever the above condition is satisfied then the decoding error, \begin{align} P_{e} \leq \delta'''({\bar n}) \end{align} where $\delta''' \rightarrow 0$ as ${\bar n} \rightarrow \infty$.

Thus as ${\bar n} \rightarrow \infty$, the combined error probability $P_{c}+P_{d}+P_{e} \rightarrow 0$ whenever the rate constraint (\ref{eq:rate constraint}) is met. Thus the bound on RHS of (\ref{eq:converse-rateconstraint}) is achievable.
\end{proof}


\subsection{Uncoordinated FDMA with Equal Allocation}
Using \eqref{eq:epsilon_uc}, we now derive the outage function for uncoordinated FDMA with equal bandwidth allocation. We assume that the users choose one of the $B$ subbands randomly. Since the aggregate arrivals are modeled as Poisson and each user chooses a subband randomly, the arrival process in each subband can also be modeled as Poisson with an appropriately scaled arrival rate~\cite{KinB1993}. If a subband is chosen by more than two users, the transmission of all those devices is assumed to be unsuccessful. Equivalently, each device chooses to transmit on a randomly selected subband using the slotted aloha protocol. Due to this, we will henceforth refer to uncoordinated FDMA as aloha FDMA. For a given common rate the transmission is successful only when the following two conditions are successful: i) the user under consideration is the only one to choose a particular subband, and ii) its channel is sufficiently strong to achieve common rate $\ncalR$ over the chosen band. For notational simplicity, we specialize the expression of common rate $\widetilde{\ncalR}_f$ derived in Lemma~\ref{lem:R_fdma} for a single user with channel gain $g_k$ as follows:
\begin{align}
\widetilde{\ncalR}_f(g_k, \mu, B) = \frac{1}{B} \log_2 (1+B\mu g_k).
\label{eq:R_f_uc}
\end{align}
Denoting by $\Pi_{UF}$ the set of aloha FDMA strategies, the outage function can now be derived as follows:


\begin{align}
&1 - \ncalE^{\Pi_{UF}}(x,\mu,\ncalR) 
 = \frac{\max_{\Pi_{UF}} {\nbbE}_{\ntKmslot,\left\{g_k \right\}}\left(\ntKsmslot\right)}{x \ntaum} \nonumber \\
&\stackrel{(a)}{=} \frac{\max_B {\nbbE}_{\ntKmslot,\left\{g_k \right\}} \left\{ {\widetilde B} \times {\nb1}\left[{\widetilde \ncalR}_f(g_k,\mu,B) \geq \ncalR \right] \right\}}{x \ntaum} \nonumber \\
&\stackrel{(b)}{=} \frac{\max_B \left\{ {\mathbb E}_{\ntKmslot} [{\widetilde B}] \times {\mathbb E}_{\{g_k\}}{\nb1}\left[{\widetilde \ncalR}_f(g_k,\mu,B) \geq \ncalR\right]\right\}}{x \ntaum}\nonumber \\
&\stackrel{(c)}{=} \frac{\max_B  \left\{\lambda \exp(-\lambda/B) \times \nbbP \left[{\widetilde \ncalR}_f(g_k,\mu,B) \geq \ncalR \right]  \right\}}{x \ntaum}\nonumber \\
&\stackrel{(d)}{=} \frac{\max_B  \left\{\lambda \exp(-\lambda/B) \times \nbbP\left[g_k \geq \frac{2^{B\ncalR}-1}{B\mu}\right]  \right\}}{x \ntaum}
\label{eq:fdmaoutage_uc}
\end{align}
where ${\widetilde B}$ in $(a)$ is the number of bins with exactly one arrival, which is a function of total number of arrivals in the minislot $\ntKmslot$ and the number of partitions $B$, $(a)$ follows from the fact that for a successful transmission, the user should arrive in one of the ${\widetilde B}$ bins and should have a strong enough channel to achieve rate $\ncalR$, $(b)$ follows from the independence of these two events, $(c)$ follows from the fact that the number of arrivals in each subband is Poisson distributed as stated above, and $(d)$ follows from \eqref{eq:R_f_uc}.

\section{Numerical Results
} \label{sec:numresults}


For the numerical results, we assume that the reference $\snr$ is $\mu = 0$dB. This would be obtained, for example, with a device transmitting with 10dBm (10mW) power over 1MHz bandwidth, a noise power spectral density of -174dBm/Hz, a receiver noise figure of 5dB, a receiver antenna gain of 14dB, a 3.76 pathloss exponent, a 128dB pathloss intercept at 1000m, and a cell radius (reference distance) of 1360m \cite{HuaPapB2012}. Assume $\ntaus=1$sec. In Figs.~\ref{fig:OutVsRate} through \ref{fig:TpVsLam}, we consider coordinated optimal, coordinated FDMA, uncoordinated optimal, and aloha FDMA (uncoordinated FDMA) transmission strategies. For uncoordinated transmission, we consider   three cases characterized by the parameter sets $\{M=1,Z=1\}$, $\{T_w = 5, M=10, Z=10\}$, and $\{T_w = 5, M=20, Z=10\}$. The first case correspond to a single transmission (i.e., no retransmissions). For the second and third cases, if every transmission is in outage, the average number of transmissions is $2M/(T_w+1)$. So the second case corresponds to approximately 4 transmissions, and the third case corresponds to approximately 8 transmissions. We refer to these cases respectively as ``1 tx," ``4 tx," and ``8 tx." 

Fig.~\ref{fig:OutVsRate} shows the outage probability versus the user rate ${\cal R}$ for a given arrival rate of $\lambda = 16$ and reference $\snr$ $\mu = 0$dB. The outages for the coordinated optimal, coordinated FDMA, uncoordinated optimal, and aloha FDMA strategies are given respectively by~\eqref{eq:optoutage_c},~\eqref{eq:fdmaoutage_c},~\eqref{eq:optoutage_uc},~\eqref{eq:fdmaoutage_uc}.  As expected, for a given outage probability, the coordinated FDMA strategy achieves a higher rate than the aloha FDMA strategy, and the coordinated optimal strategy achieves a higher rate than the uncoordinated optimal strategy. For the uncoordinated optimal strategy, the single transmission case is uniformly the best, i.e., for any outage probability, the corresponding rate for 1 tx is the highest. For the aloha FDMA strategy, the best average number of transmissions depends on the desired outage probability.  In the regime of interest where the outage probability is 0.1, the best aloha FDMA rate is achieved with 4 tx, and the coordinated FDMA rate is about a factor 4 greater than this rate. On the other hand, the coordinated optimal user rate is only about a factor of 1.2 greater than the best (1 tx) uncoordinated optimal rate. 

Still assuming $\lambda = 16$ and $\mu = 0$dB, Fig.~\ref{fig:TpVsRate} shows the throughput ${\cal S}(\lambda, \mu, {\cal R})$ versus user rate ${\cal R}$ for the four transmission options computed using \eqref{eq:S_c_x} for the coordinated strategies and \eqref{eq:S_uc_x} for the uncoordinated strategies.  For the case of the 1 tx uncoordinated optimal transmission, the maximum throughput is achieved when ${\cal R} = 0.4$ and the outage, obtained from Fig.~\ref{fig:OutVsRate}, is $\epsilon = 0.20$. For the other strategies, the maximum throughput is achieved at even higher outages. For coordinated optimal transmission, it is $\epsilon = 0.71$. For the 1-tx aloha and coordinated FDMA transmissions, the respective outages are $\epsilon = 0.50$ and $\epsilon = 0.89$ respectively. Because a practical system would not tolerate such high outages, we are motivated to consider the throughputs subject to a maximum outage constraint, which we do next.


Figure \ref{fig:TpVsLam} shows the maximum throughput (maximized with respect to the user rate) versus arrival rate $\lambda$ under the coordinated strategies when the outage probability is limited to $\epsilon_{\max} = 0.1$ \eqref{eq:optimization_c}  and under uncoordinated strategies when the failure probability is limited to $\delta_{\max} = 0.1$ \eqref{eq:optimization_uc}. The throughput increases linearly with respect to the log arrival rate. The throughput of the optimal transmission strategies seem to exhibit the same slope regardless of whether they are coordinated or uncoordinated, and the throughput slope of the optimal strategies are steeper than the FDMA strategies.  Among the uncoordinated optimal strategies, the 1 tx option is the best, but among the aloha FDMA strategies, the 4  tx option is the best. We note that the throughput spectral efficiency is relatively high because each device transmits with full power 10mW resulting in a fairly high total transmit power.  

\begin{figure}[t]
\centering
\includegraphics[width=1\columnwidth]{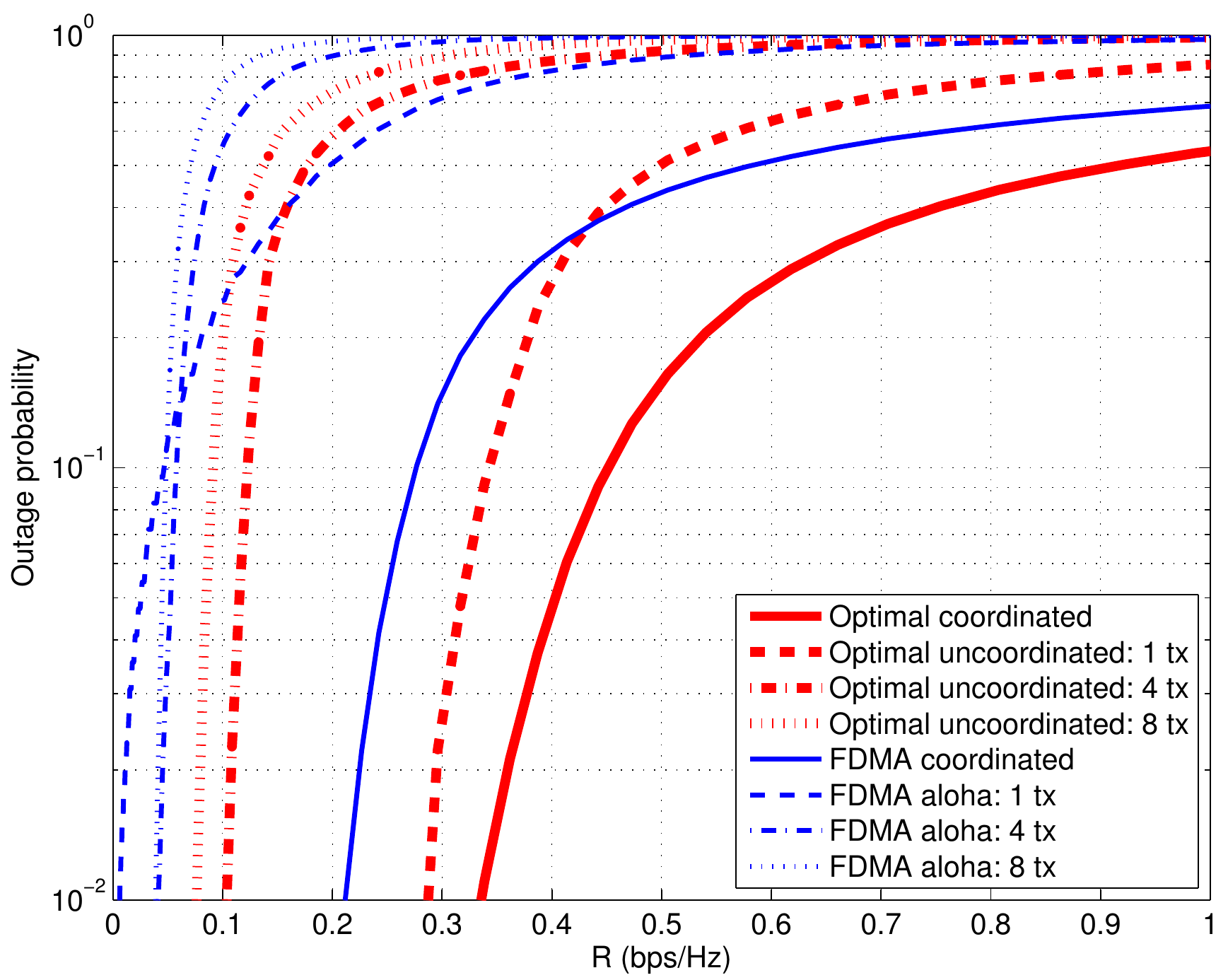}
\caption{Outage probability vs. common rate $\ncalR$ ($\lambda=16$).}
\label{fig:OutVsRate}
\end{figure}

\begin{figure}[t]
\centering
\includegraphics[width=1\columnwidth]{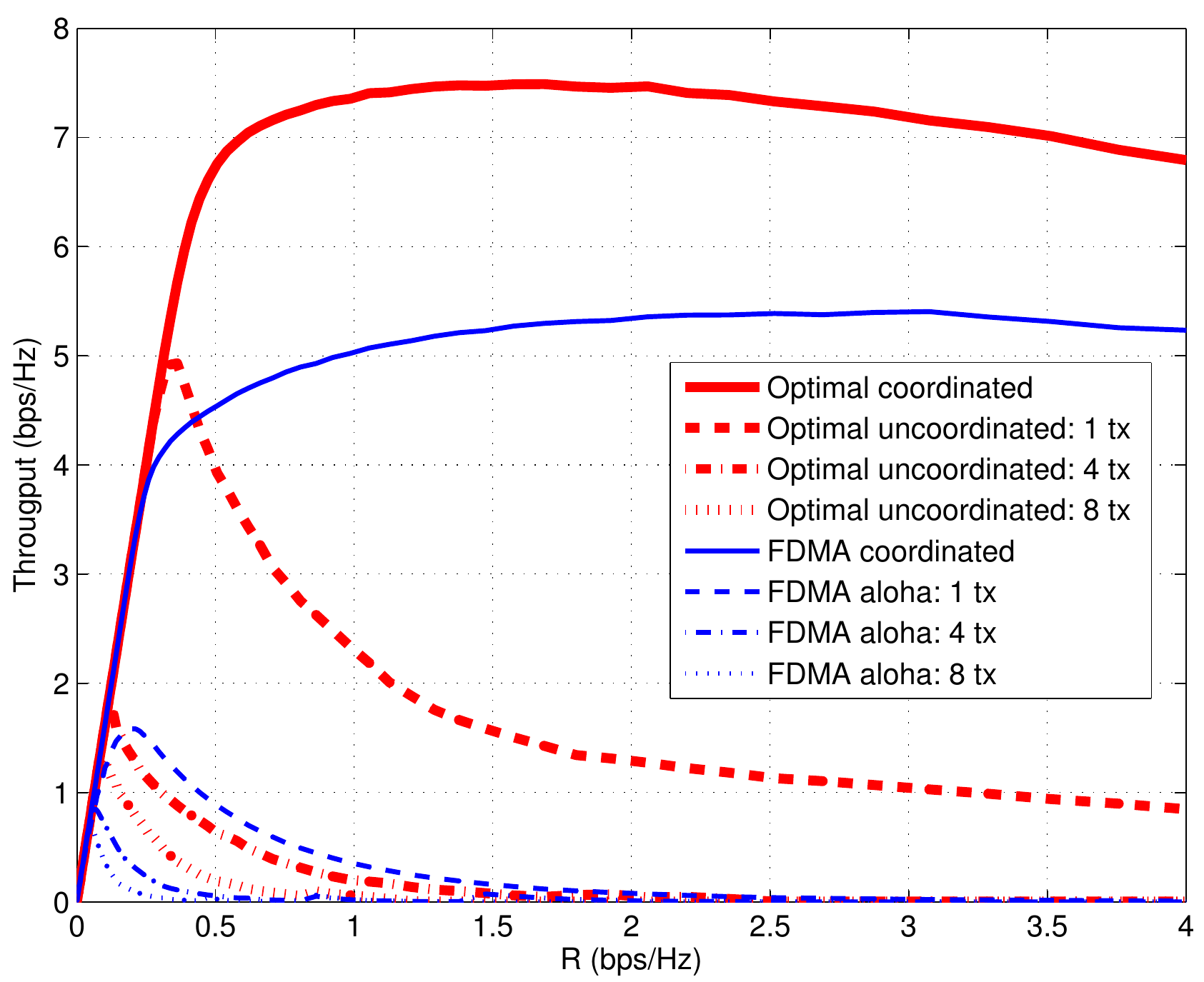}
\caption{Throughput vs. common rate $\ncalR$. The throughputs for coordinated and uncoordinated strategies are given by~\eqref{eq:S_c_x} and~\eqref{eq:S_uc_x}, respectively ($\lambda=16$).}
\label{fig:TpVsRate}
\end{figure}


\begin{figure}[t]
\centering
\includegraphics[width=1\columnwidth]{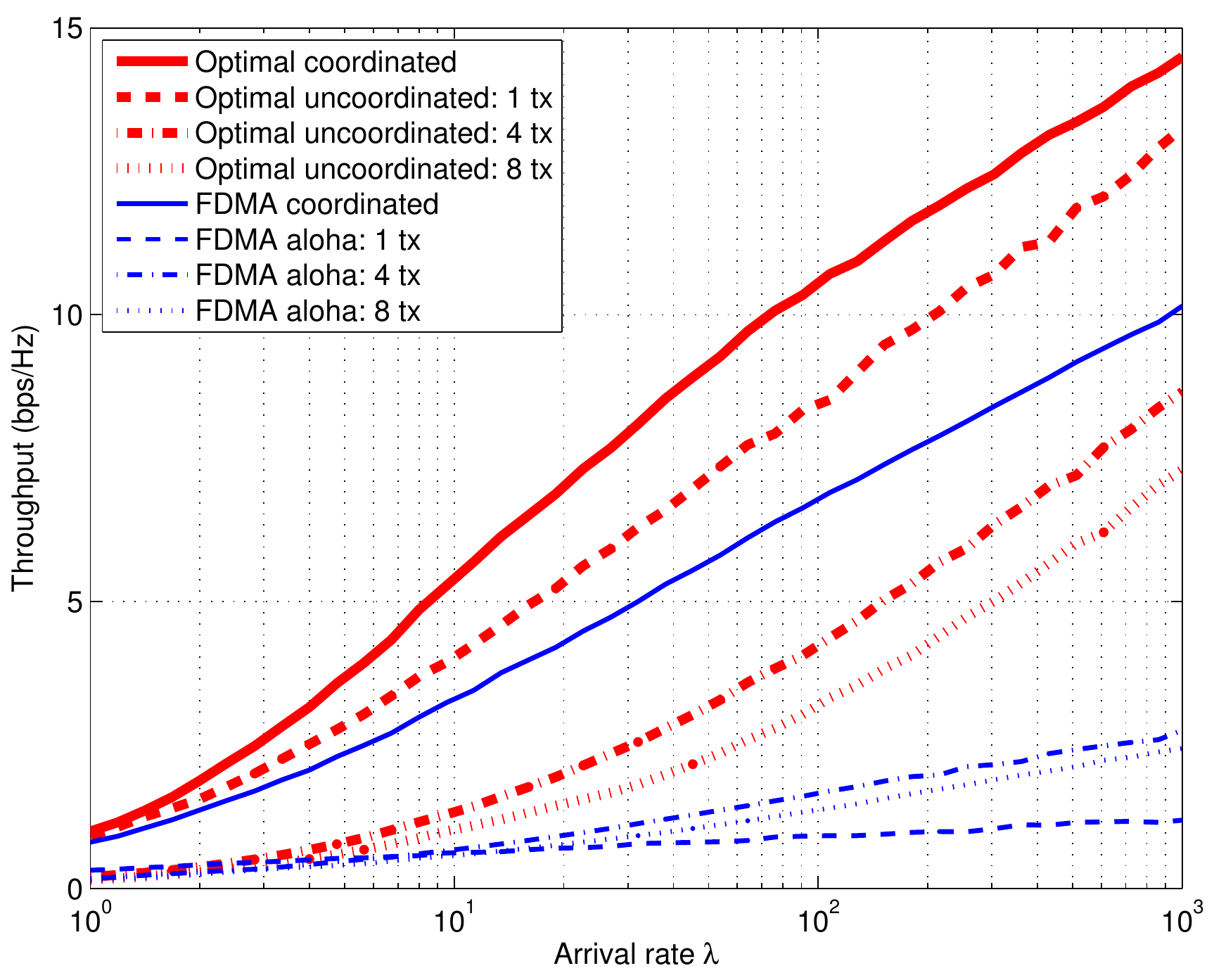}
\caption{Maximum throughput vs. arrival rate, subject to a maximum outage or failure probability constraint of $10\%$.}
\label{fig:TpVsLam}
\end{figure}


\section{One-Stage Versus Two-Stage Transmission}
In the previous sections, we analyzed the performance of coordinated and uncoordinated transmission strategies for multiple access. In this section we build upon these results to study the performance of more complete communication protocols that account for overhead signaling such as acknowledgements and scheduling information. 

In a conventional protocol for scheduled uplink communication in a cellular network, for example based on the LTE standard, multiple stages of uplink and downlink communication are used to establish the identity of a user, to resolve collisions on the random access channel, and to provide scheduling information. Following the handshaking mechanism, the data payload is transmitted on scheduled uplink resources. If the resources required for transmitting the data payload are large compared to the overhead resources, then it is worthwhile to invest in the overhead to ensure efficient communication of the payload over contention-free channels. On the other hand, if the data payload is small it may not be worthwhile to make this investment. Therefore, we are interested in studying a simpler protocol where the data payload and user identity are transmitted over a random access channel without any prior overhead. 
\begin{figure}[t]
\centering
\includegraphics[width=.85\columnwidth]{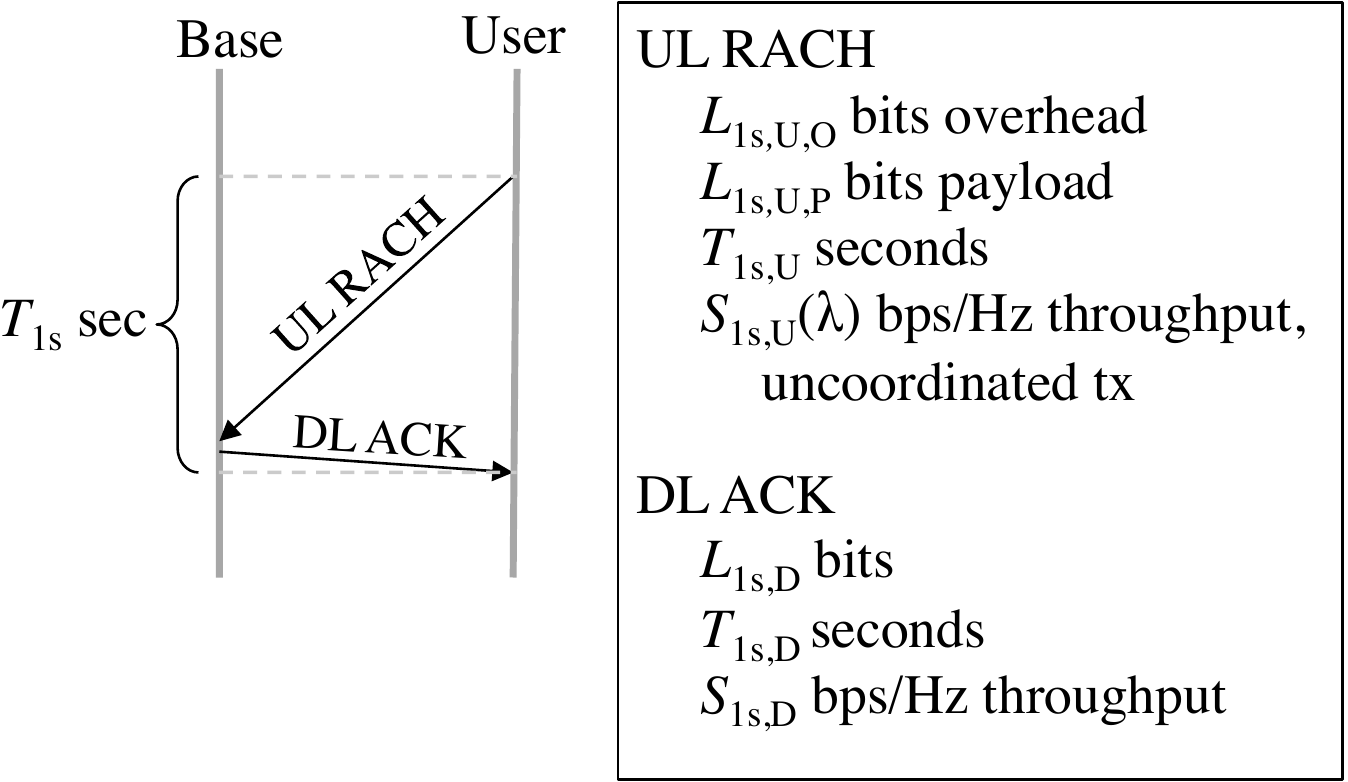}
\caption{One-stage protocol.}
\label{fig:1StageProtocol}
\end{figure}

\begin{figure}[t]
\centering
\includegraphics[width=.85\columnwidth]{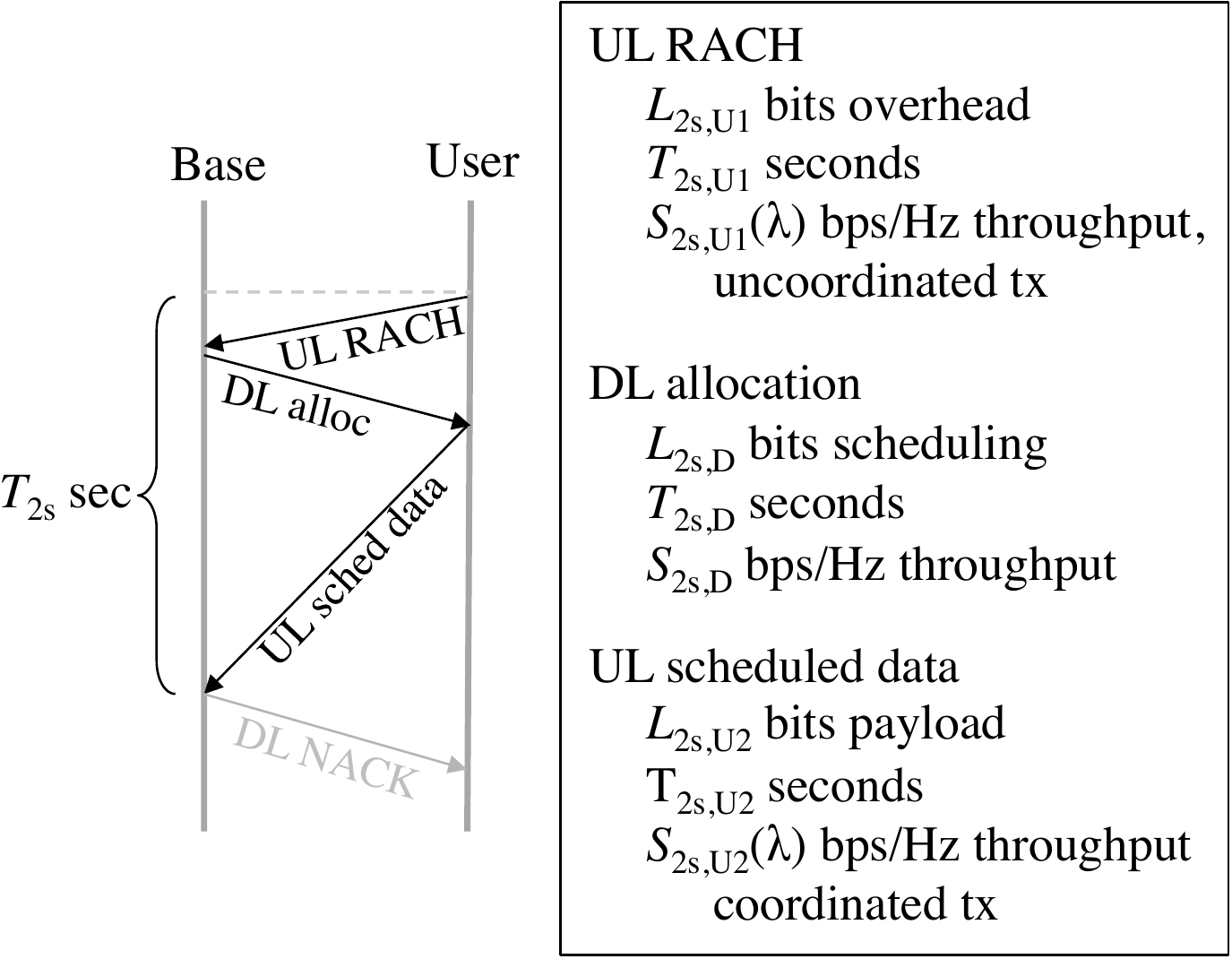}
\caption{Two-stage protocol.}
\label{fig:2StageProtocol}
\end{figure}

This simpler protocol is shown in Fig. \ref{fig:1StageProtocol}. We call this the {\em one-stage} protocol because there is a single stage of uplink transmission which occurs over a random access channel. Its performance can be characterized by the uncoordinated strategies discussed in Section \ref{sec:uncoordinated}. The protocol for scheduled data transmission, which we call the {\em two-stage} protocol, is shown in Fig. \ref{fig:2StageProtocol}. We simplify the multiple stages of the handshaking to just two stages, where the overhead is sent on the first uplink stage, and the data payload is sent on the second. Transmission on the first stage is over a random access channel while the transmission on the second stage is over scheduled resources. The performance of the two stages are characterized respectively by the uncoordinated and coordinated strategies (Section \ref{sec:coordinated}). Additional stages could have been considered, but this simplified framework allows for a more straightforward comparison with the one-stage protocol. We assume in general that users are synchronized so that the transmissions on the random access channels arrive at the base station at the beginning of a slot epoch (or minislot epoch when retransmissions are considered). The average number of users per slot of duration 1 second ``arriving" with data to transmit is $\lambda$. We do not account for the overhead of a broadcast downlink synchronization signal.

Under the one-stage protocol, users transmit $L_{1s,U}$ bits over the uplink random access channel. The bits consist of the data payload and the information for uniquely identifying the user. Given the bandwidth $W$ Hz and letting $T_{1s,U}$ be the transmission time in seconds, the throughput spectral efficiency for an average arrival rate $\lambda$ is upper bounded by the maximum throughput $\ncalS_{1s,U}$  (\ref{eq:S_uc_x})
\begin{equation}
\label{eq:1stageUl}
\frac{\lambda (L_{1s,U,O} + L_{1s,U,P})}{WT_{1s,U}} \leq \ncalS_{1s,U}(\lambda).
\end{equation}
If a user's packet is successfully decoded, the base station sends a positive acknowledgement consisting of $L_{1s,D}$ bits to the user over $T_{1s,D}$ seconds. We assume the downlink spectral efficiency $\ncalS_{1s,D}$ bps/Hz is independent of the user load. For a fixed downlink payload $L_{1s,D}$ and the downlink rate distribution among randomly, uniformly dropped users in a disk around the base generated according to $\log_2(1+\mu_{r(k)})$ (where $\mu_{r(k)}$ is the received $\snr$ of user $k$ located at a distance $r$ from the base station, see \eqref{eq:SNR}), the average spectral efficiency is the harmonic mean of the rates: $\ncalS_{1s,D} = \lim_{K \rightarrow \infty} K\left[ \sum_k \log_2(1_\mu{r(k)} \right]^{-1}$, where $K$ is the total number of users. Therefore,
\begin{equation}
\label{eq:1stageDl}
\frac{\lambda L_{1s,D}}{WT_{1s,D}} \leq \ncalS_{1s,D}(\lambda).
\end{equation}
Combining (\ref{eq:1stageUl}) and (\ref{eq:1stageDl}), and defining $[x]^+ = \max[x,0]$, an upper bound on the uplink payload $L_{1s,U}$ as a function of the capacity $\lambda$ is
\begin{align}
\label{eq:1stageLamVsL}
L_{1s,U,P} &\leq \left[\ncalS_{1s,U}(\lambda) \frac{WT_{1s,U}}{\lambda} - L_{1s,U,O} \right]^+ \nonumber \\
&\leq \left[\ncalS_{1s,U}(\lambda) \left(\frac{WT_{1s}}{\lambda} - \frac{D_{1s,D}}{\ncalS_{1s,D}} \right) - L_{1s,U,O}\right]^+.
\end{align}

In the two-stage strategy, users vie for attention from the base station by transmitting $L_{2s,U1}$ bits over a random access channel. For the set of successfully detected users, the base station determines how the uplink resources should be allocated for scheduled payload transmission. For each detected user, the base station conveys a downlink control message of $L_{2s,D}$ bits to indicate the allocated resource. Then these users each transmit $L_{2s,U2}$ bits of payload in a coordinated manner over the resources. Because users are assumed to have received the downlink resource allocation message correctly, a positive acknowledgement is not required from the base station after correctly demodulating the data payload. A negative acknowledgement is sent in the case it is not correctly received, but since this event will be rare, we do not account for its resources. For the uplink random access channel (RACH), downlink transmission, and uplink data transmission, we have the following inequalities
\begin{align}
\frac{\lambda L_{2s,U1}}{WT_{2s,U1}} &\leq \ncalS_{2s,U1}(\lambda) \\
\frac{\lambda L_{2s,D}}{WT_{2s,D}} &\leq \ncalS_{2s,D} \\
\frac{\lambda L_{2s,U2}}{WT_{2s,U2}} &\leq \ncalS_{2s,U2}(\lambda).
\end{align}
The uplink RACH throughput is given by the uncoordinated throughput (\ref{eq:S_uc_x}), and the scheduled uplink throughput is given by the coordinated throughput (\ref{eq:S_c_x}). Combining these inequalities yields the following relationship between the payload $L_{2s,U2}$ and arrival rate $\lambda$
\begin{align}
\label{eq:2stageLamVsL}
L_{2s,U2} &\leq \ncalS_{2s,U2}(\lambda) \frac{WT_{2s,U2}}{\lambda}  \nonumber \\
&\leq \left[\ncalS_{2s,U2}(\lambda) \left(\frac{WT_{2s}}{\lambda} - \frac{L_{2s,U1}}{\ncalS_{2s,U1}(\lambda)} - \frac{L_{2s,D}}{\ncalS_{2s,D}} \right) \right]^+. 
\end{align}

\begin{figure}[t]
\centering
\includegraphics[width=1\columnwidth]{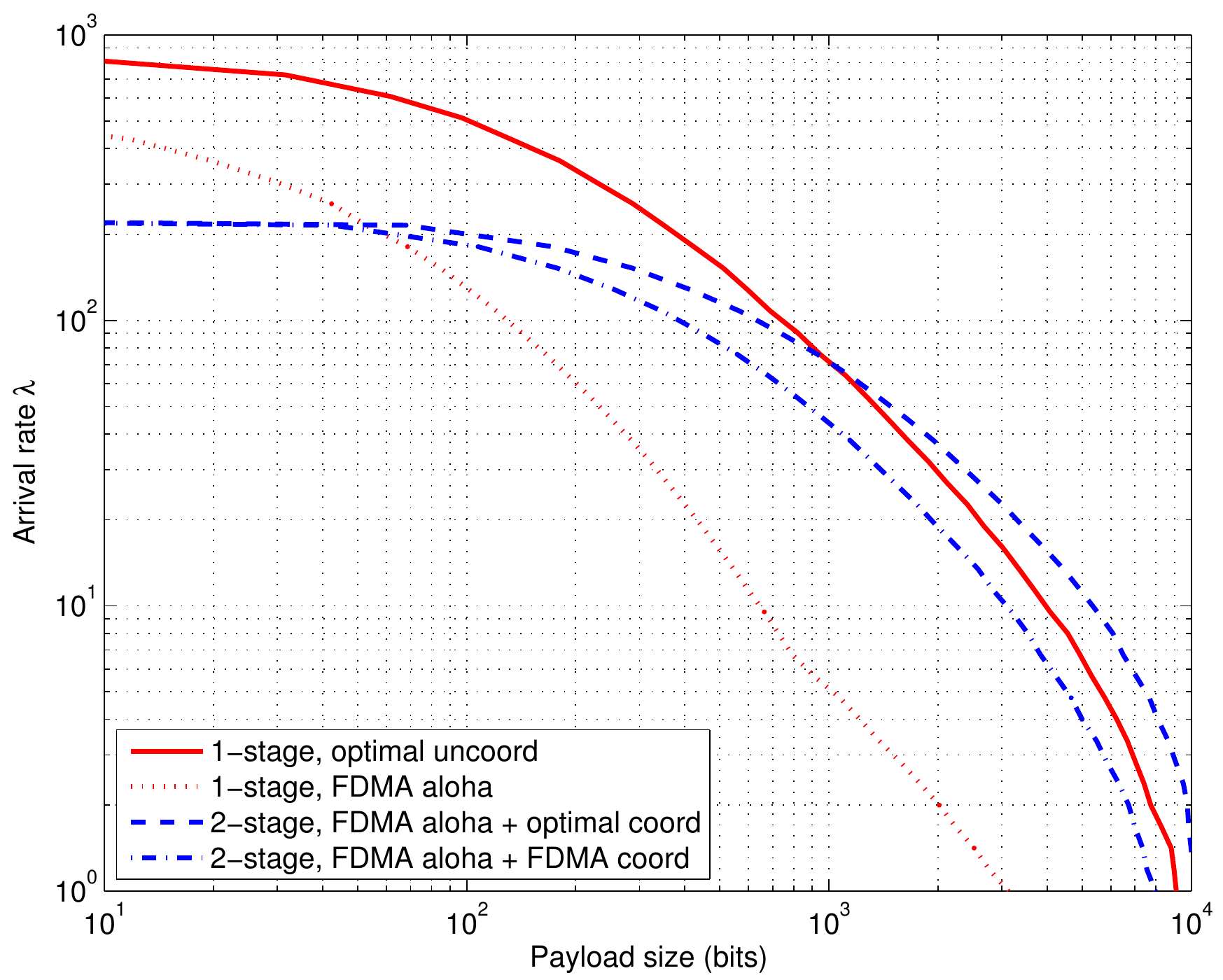}
\caption{Maximum supportable arrival rate versus payload size assuming $W=10$ KHz bandwidth resources and $1$ second latency constraint.}
\label{fig:Capacity1stageVs2stage}
\end{figure}

For the numerical results, we assume $W = 10$KHz of bandwidth resources, a reference $\snr$ of $\mu = 0$dB, and a latency constraint of 1sec. Using overhead parameters motivated by LTE, we assume 20 bits are used for ID, 64 bits are used for downlink scheduling, and 20 bits are used for positive acknowledgements. Therefore $L_{1s,U,O} = 20$, $L_{1s,D} = 20$, $L_{2s,U1} = 20$, $L_{2s,D} = 64$. For the one-stage strategy, we consider either uncoordinated optimal or aloha FDMA transmission. For the two-stage strategy, we assume aloha FDMA for the first stage and coordinated optimal or coordinated FDMA for the second stage. Assuming a downlink reference $\snr$ of 0dB, the downlink spectral efficiency computed from the harmonic mean of rates is $\ncalS_{2s,D} = 2.07$ bps/Hz.

Fig. \ref{fig:Capacity1stageVs2stage} shows the supportable arrival rate $\lambda$ versus the payload size as given by (\ref{eq:1stageLamVsL}) for the one-stage performance and (\ref{eq:2stageLamVsL}) for the two-stage performance. For smaller payload sizes, the supportable arrival rate for the one-stage strategies is higher than the arrival rate of the two-stage strategies because the two-stage overhead outweighs the relative inefficiency of the uncoordinated transmission. For larger payloads, the overhead becomes negligible, and the two-stage strategies are relatively more efficient. The crossover threshold between the one-stage optimal and two-stage with coordinated optimal second stage is about 1000 bits, so that the one-stage strategy supports a higher arrival rate for payloads smaller than this threshold. For a payload size of 100 bits, the supportable arrival rate for one-stage optimal is about 2.5 times that of the two-stage strategies. The crossover threshold between the one-stage aloha FDMA and two-stage strategies is about 60 bits. On the other hand, for larger payload sizes, the two-stage strategies are far superior compared to one-stage FDMA, and they achieve similar performance compared to one-stage optimal. 


\section{Conclusion}
In this paper, we developed a systematic framework to study the throughput optimal system design for randomly arriving M2M devices in a cellular uplink. Using a novel analytic framework, we characterized the maximum throughput achievable under a variety of coordinated and uncoordinated strategies, with the latter class of strategies having a provision for retransmissions under a maximum latency constraint. Incorporating in a novel way the joint decoding techniques of coordinated multiple access channels to random access, we obtained a new fundamental result characterizing the throughput performance of optimal uncoordinated random access transmission. Using these results, we perform a realistic comparison of a {\em one-stage} design, where the data payload is communicated through random access, and {\em two-stage} design, where the uplink connection is established though random access in the first stage and data payload is communicated over contention-free resources in the second stage. Our analysis concretely demonstrates that for payloads of 1000 bits or less, the  optimal one-stage design supports more devices than the two-stage design due to the reduced overhead. 

There are numerous extensions possible for this work. From system implementation perspective, it is important to understand whether it is preferable to implement optimal uncoordinated strategy proposed in this paper given its relatively high complexity and sensitivity to practical impairments such as non-ideal channel estimation. From information theory perspective, it is important to extend the analysis of optimal uncoordinated strategy to the case where each user has knowledge of its own channel and can thus perform power control in a distributed way. From cellular systems perspective, it is important to extend this study to multi-cell scenarios.

\appendices
\section{Proof of Theorem~\ref{thm:Delta} \label{App:Delta}}
Let $X_i$ be the back-off time after $i^{th}$ transmission. For notational simplicity let $X_0 = 1$ be the time of the first transmission. Now define the event $\mathcal{A}_n$ that there are exactly $n$ transmission attempts before the deadline is expired. It can be mathematically expressed as
\begin{align}
\mathcal{A}_n = \sum_{i=0}^{n-1} X_i \leq M, \sum_{i=0}^{n} X_i > M.
\end{align}
Now, we calculate the probability of these events, starting with $\mathcal{A}_1$ as follows
\begin{align}
\nbbP\left[\mathcal{A}_1 \right] &= \nbbP\left[   X_0 \leq M, X_0 + X_1 > M \right]
\stackrel{(a)}{=} \nbbP\left[X_1 > M-1 \right] \nonumber \\
&= 1 - \sum_{i=1}^{M - 1} \nbbP[X_1 = i]
= 1 - \frac{1}{T_w}\sum_{i=1}^{M - 1} \mathbf{1}(i \in \nbbZ_{1}^{T_w}) \label{eq:delta_A1}\nonumber \\
&= \left\{ \begin{array}{ll}
1 - \frac{M-1}{T_w}; & M < T_w + 1\\
0 & M \geq T_w + 1\\
\end{array} \right.
\end{align}
where $(a)$ follows from the fact that $X_0 = 1$ and $M \geq 1$ by assumption. Now the probability of the events $\mathcal{A}_n$, $1<n< Z$, can be calculated as follows
\begin{align}
\nbbP\left[\mathcal{A}_n \right] &= \nbbP\left[   \sum_{i=0}^{n-1} X_i \leq M, \sum_{i=0}^{n} X_i > M \right]\nonumber \\
&= \nbbP\left[   \sum_{i=1}^{n-1} X_i \leq M-1, \sum_{i=1}^{n} X_i > M-1 \right] \nonumber \\
&= \nbbP\left[M - 1 - X_n < S_{n-1} \leq M -1 \right]\nonumber\\
&= \frac{1}{T_w} \sum_{j=1}^{T_w} \nbbP\left[M - 1 - j < S_{n-1} \leq M -1 \right]\nonumber\\
&= \frac{1}{T_w} \sum_{j=1}^{T_w} \sum_{k=M - j}^{M - 1} \nbbP[S_{n-1} = k].
\end{align}
Using Lemma~\ref{lem:UnifSum}, it can be expressed as
\begin{align}
\frac{1}{T_w} \sum_{j=1}^{T_w} \sum_{k=M - j}^{M - 1} \mathbf{1} \left( k \in \nbbZ_{n-1}^{(n-1)T_w} \right) \frac{1}{T_w^{n-1}}  {n-1 \choose k-n}_{T_w}.
\label{eq:delta_An}
\end{align}
On the similar lines, the probability of the event $\ncalA_Z$ can be calculated as follows
\begin{align}
\nbbP\left[{\ncalA_Z}\right] &= \nbbP \left[ \sum_{i=0}^{Z-1} X_i \leq M \right]
= \nbbP \left[S_{Z-1} \leq M - 1 \right]\nonumber\\
&= \sum_{j=0}^{M - Z} \nbbP \left[S_{Z-1} = Z-1 + j  \right]\nonumber\\
&= \sum_{j=0}^{M-Z} \nb1 \left( j \in \nbbZ_{0}^{(Z-1)(T_w-1)}\right) \frac{1}{(T_w)^{Z-1}} {Z-1 \choose j}_{T_w},
\label{eq:delta_AZ}
\end{align}
where the last step follows from Lemma~\ref{lem:UnifSum}. The final result now follows from \eqref{eq:delta_A1}, \eqref{eq:delta_An} and \eqref{eq:delta_AZ} along with the fact that $\delta$ can be expressed in terms of $\{\ncalA_n\}$ as
\begin{align}
\delta &= \sum_{n=1}^Z \epsilon^n \nbbP[\mathcal{A}_n].
\label{eq:Delta_intermed}
\end{align}
This completes the proof. \hfill \QED


\bibliographystyle{IEEEtran}
\bibliography{M2M_Journal_v2.1}

\end{document}